\documentclass[11pt]{article} 
\usepackage[sectionbib]{natbib}
\usepackage{array,epsfig,fancyheadings,rotating,setspace}
\usepackage[]{hyperref}  
\usepackage{sectsty, secdot}
\sectionfont{\fontsize{12}{14pt plus.8pt minus .6pt}\selectfont}
\renewcommand{\theequation}{\thesection\arabic{equation}}
\subsectionfont{\fontsize{12}{14pt plus.8pt minus .6pt}\selectfont}
\usepackage[lmargin=1in,rmargin=1in,tmargin=1in]{geometry}

\RequirePackage{graphicx,subfigure,latexsym,amssymb}
\RequirePackage{float,epsfig,multirow,rotating,times}
\RequirePackage{upgreek,wrapfig}
\usepackage{amsmath}
\usepackage{amssymb}
\usepackage{amsfonts}
\usepackage{multirow}
\usepackage{amsthm}
\usepackage{algorithm}
\usepackage[noend]{algpseudocode}
\usepackage{mathrsfs}
\usepackage{epstopdf}
\newtheorem{theorem}{Theorem}
\newtheorem{lem}{Lemma}
\newtheorem{corollary}{Corollary}
\newtheorem{prop}{Proposition}

\def\argmax{\mathop{\rm argmax}}

\newcommand{\s}{\ensuremath{\mathbb{S}}}

\newcommand{\real}{\ensuremath{\mathbb{R}}}

\newcommand{\ltwo}{\ensuremath{\mathbb{L}^2}}
\newcommand{\lone}{\ensuremath{\mathbb{L}^1}}
\newcommand{\linf}{\ensuremath{\mathbb{L}^{\infty}}}

\newcommand{\inner}[2]{\left\langle #1,#2 \right\rangle}

\def\T{{ \mathrm{\scriptscriptstyle T} }}
\newcommand{\abs}[1]{\left\vert#1\right\vert}

\newcommand{\norm}[1]{\left\Vert#1\right\Vert}

\usepackage{psfrag}
\usepackage{mathrsfs}
\usepackage{accents}
\usepackage{authblk}
\usepackage{blindtext}

\begin{document}


\markright{ \hbox{\footnotesize\rm 
}\hfill\\[-13pt]
\hbox{\footnotesize\rm
}\hfill }

\markboth{\hfill{\footnotesize\rm Dasgupta S., Pati D. and Srivastava A.} \hfill}
{\hfill {\footnotesize\rm Dasgupta S., Pati D. and Srivastava A.} \hfill}

\renewcommand{\thefootnote}{}
$\ $\par


\fontsize{12}{14pt plus.8pt minus .6pt}\selectfont \vspace{0.8pc}
\centerline{\large\bf A Two-Step Geometric Framework For Density Modeling}
\vspace{.4cm} \centerline{Sutanoy Dasgupta, Debdeep Pati, and Anuj Srivastava} \vspace{.4cm} \centerline{\it
Department of Statistics, Florida State University} \vspace{.55cm} \fontsize{9}{11.5pt plus.8pt minus
.6pt}\selectfont


\begin{quotation}
\noindent {\it Abstract:}
We introduce a novel two-step approach for estimating a probability density function ({\it pdf}) given its samples, 
with the second and important step coming from a geometric formulation. 
The procedure involves obtaining an initial estimate of the {\it pdf}
and then transforming it via a warping function to reach the final estimate. 
The initial estimate is intended to be computationally fast, albeit suboptimal, but its warping
creates a larger, flexible class of density functions, resulting in substantially improved estimation.  The search for 
optimal warping is accomplished by 
mapping diffeomorphic functions to the tangent space of a Hilbert sphere, a vector space whose elements 
can be expressed using an orthogonal basis. Using a truncated basis expansion, we 
estimate the optimal warping under a (penalized) likelihood criterion and, thus, the optimal density estimate. This framework is introduced for univariate, 
unconditional {\it pdf} estimation and then extended to conditional {\it pdf} estimation.
The approach avoids many of the computational pitfalls associated with classical conditional-density estimation methods, without losing on estimation performance.
We derive asymptotic convergence rates of the density estimator and demonstrate this approach using 
both synthetic datasets and real data, the latter relating to the association of a toxic metabolite on preterm birth.

\vspace{9pt}
\noindent {\it Key words and phrases:}
conditional density; density estimation; warped density; Hilbert sphere; sieve estimation; tangent space; weighted likelihood maximization
\par
\end{quotation}\par

\def\thefigure{\arabic{figure}}
\def\thetable{\arabic{table}}

\renewcommand{\theequation}{\thesection.\arabic{equation}}

\fontsize{12}{14pt plus.8pt minus .6pt}\selectfont


\setcounter{equation}{0} 
\section{Introduction}

Estimating a probability density function ({\it pdf}) 
 is an important and well studied field of research in statistics. The most basic problem in this area is that of univariate {\it pdf} estimation from {\it iid} samples, henceforth referred to as unconditional density estimation.
Another problem of significance is conditional density estimation. Here one needs to characterize the behavior of the response variable for 
different values of the predictors. 

Given the importance of {\it pdf} estimation in statistics and related disciplines, 
a large number of solutions have been proposed for each of these problems. While the earliest works focused on parametric solutions,  
the trend over the last three decades has been to use a nonparametric approach as it
minimizes making assumptions about the underlying density (and the relationships 
between variables for conditional and joint densities). 
The most common nonparametric techniques are kernel based; 
please refer to \citet{rosenblatt1956remarks,hall1991optimal,sheather1991reliable,li2007nonparametric} for a narrative of works. Related to these approaches are ``tilting'' or ``data sharpening'' techniques for unconditional density estimation, 
see for example \citet{hjort1995nonparametric,doosti2016making}, 
and the references therein. Kernel methods are very powerful in univariate setting. However, as the number of variables involved gets higher, 
these methods tend to be computationally inefficient because of the complexities involved in bandwidth selection, especially in conditional density estimation setup.

\subsection{Two-Step Approaches for Density Estimation}
Another common approach for {\it pdf} estimation, and the one pursued in the current paper, 
is a two-step estimation procedure discussed in \citet{leonard1978density,lenk1988logistic,lenk1991towards,tokdar2010bayesian,tokdar2007towards}, etc.
In the first step,  one estimates
an initial {\it pdf}, say $f_p$, from the data, perhaps restricting to a parametric family.
Then, in the second step, one {\it improves} upon this 
estimate by forming a function $w> 0$, 
that depends on the initial estimate $f_p$, and forming a final estimate using $w(x) f_p(x) /\int_y w(y) f_p(y)yd$. 
Thus, the second step involves estimation of an optimal $w$ in order to estimate the overall {\it pdf}.
In a Bayesian context, the function $w$ is  often assigned a Gaussian process prior.  
While this approach is quite comprehensive,  
the calculation of the normalization constant makes the computation very cumbersome. 
The two-step procedures can also be adapted for estimating conditional density functions: first estimate 
the conditional mean function and then estimate the conditional density of the residuals, as is done in \citet{hansen2004nonparametric}. 
Over the recent years, Bayesian methods for estimating {\it pdfs} based on mixture models and latent variables have received a lot of attention, 
primarily due to their excellent practical performances and an increasingly rich 
set of algorithmic tools for sampling posterior using Markov Chain Monte Carlo (MCMC) methods. References include \citet{escobar1995bayesian,muller1996bayesian,maceachern1998estimating,kalli2011slice,jain2012split,kundu2014latent,bhattacharya2010latent} among others.  
 However, these results also come at a very high computational cost typically associated with the MCMC algorithms. 
 Applications of flexible Bayesian  models for conditional densities are discussed in
\citet{MacEachern:99,DeIorioMullerRosnerMacEachern:04,GriffinSteel:06,DunsonPillaiPark:07,ChungDunson:09,NoretsPelenis2012}, 
among others.  Although the literature suggests that such methods based on  mixture models have several attractive properties,  
they lack interpretability and  the MCMC solutions for model fitting are overly complicated and expensive.

\subsection{A Geometric Two-Step Approach}

In this article, we pursue a 
geometric, two-step approach 
that is applicable to both conditional and unconditional density estimation. The main 
motivation here is develop an efficient estimation procedure while retaining good estimation performance.  The 
main difference from the previously described two-step procedure is that the transformation of $f_p$ (in the second 
step) is now based on the action of a diffeomorphism group, as follows. 
Let $f_p$ be a strictly positive univariate density on the interval $[0,1]$; 
$f_p$ serves as an initial estimate of the {\it pdf}. 
Let $\Gamma$ be the set of all positive diffeomorphisms from $[0,1]$ to itself, i.e.
$\Gamma=\{\gamma | \gamma \text{ is differentiable},{\gamma}^{-1}\text{ is differentiable, }\dot{\gamma}>0,\gamma(0)=0,\gamma(1)=1\}$.
The elements of  $\Gamma$ play the role of warping functions, or transformations of $f_p$. 
Given a $\gamma \in \Gamma$, the transformation of $f_{p}$ is defined by: 
$( f_{p}, \gamma) =  (f_{p} \circ \gamma) \dot{\gamma}$.
Henceforth, this transformation is referred to as {\it warping} of $f_p$, 
and the resulting {\it pdf} $f$ as a {\it warped density}. This mapping is comprehensive in the sense that one can  
go from any positive {\it pdf} to any other positive {\it pdf} using an appropriate $\gamma$.  
Note that since $\int_{0}^1 f_{p}(\gamma(x)) \dot{\gamma}(x) dx = 1$,
there is no need to normalize this transformation. However, 
the difficulty of estimating the normalizing constant now shifts to the problem of 
estimating over $\Gamma$ and this poses some challenges as $\Gamma$ is a nonlinear
manifold. 
Note that the use of diffeomorphisms as transformations of a {\it pdf} have been used in the past, 
albeit with a different setup and scope; see, for example \citet{saoudi1994non,saoudi1997some}. 
Also, the notion of transformation 
between {\it pdf}s  has been used in the literature on {\em optimal transport} as in \citet{tabak2013family,tabak2014data}, with the difference being that  
the transport is achieved using an iterated composition of maps and not through an optimization over $\Gamma$ as done in the current paper.   
There are two parts to this paper: 
\begin{enumerate}

\item {\bf Univariate {\it pdf} Estimation}: 
We start the paper with a framework for estimating an unconditional, univariate {\it pdf} defined on $[0,1]$. 
This simple setting helps explain and illustrate the main ingredients of the framework. 
Besides, the proposed geometric framework is naturally univariate in the sense that the transformation defined 
earlier acts on univariate density shapes, making it a logical starting point for developments. 
In this simple setup, the approach delivers excellent performance while avoiding heavy computational cost, and 
is comparable to standard kernel methods, even at very low sample sizes.
The framework is then extended to univariate densities with unknown support by scaling the observation domain to $[0,1]$. 
A defining characteristic of this
warping transformation is that the initial estimate can be constructed in anyway -- parametric (e.g. gaussian) or nonparametric (e.g. kernel estimate), and is allowed to be a sub-optimal estimate of the true density. 

\item {\bf Conditional Density Estimation}: 
The second part of the article focuses on extending the framework to estimation of conditional density $f(y|x)$ from  $\{(y_i, x_i): i=1,\ldots, n,y \in \real , x \in \real^d,d\geq 1\}$.
The approach is to  start with a nonparametric mean regression model of the form 
$y_i =  m(x_i) + \epsilon_i$, $\epsilon_i \sim {\cal N}(0,\sigma^2)$, where $m(\cdot)$ is estimated using a standard 
nonparametric estimator, to obtain an initial conditional density estimate $f_{p,x} \equiv {\cal N}(\hat{m}(x),\hat{\sigma}^2)$ at the location $x$.
Then $f_{p,x}$ is warped using a warping function $\gamma_x$ into a final conditional density estimate.
Naturally, the choice of $\gamma_x \in \Gamma$ varies with the predictor $x$. 
The selection of $\gamma_x$ is based on a weighted-likelihood objective function 
that borrows information from the neighborhood of the location $x$ at which the conditional density is being evaluated.
\end{enumerate}

The main contributions of this paper as as follows: 
\begin{enumerate}
\item {\bf Avoids Normalizing Constant}: It introduces a geometric approach to two-step estimation, with the second step being based on 
the action of the diffeomoprhism group on the set of positive {\it pdfs}. This action is chosen so that one does not need a
normalization constant, and the resulting estimation process is efficient. 

\item {\bf Uses Geometry of $\Gamma$}: It uses the differential geometry of $\Gamma$ to map its elements into a subset of 
a Hilbert space, allowing for a basis expansion and application of standard optimization tools for estimating warping functions.  

\item {\bf Conditional Density Estimation}: It leads to an efficient framework for estimating conditional densities, providing very competitive practical performance and improved computational cost compared to standard kernel techniques.

\end{enumerate}

The rest of this paper is organized as follows. 
Section 2 outlines the general framework for a univariate unconditional density estimation while
Section 3 presents an asymptotic analysis of this estimator.  Section 4 contains some simulation study. Section 5 develops 
theory for conditional density estimation and illustrates properties of 
the proposed method using simulated datasets. Applications of conditional density estimation using 
the proposed framework on a real dataset are also presented.



\setcounter{equation}{0} 
\section{Proposed Framework}
In this section we develop a two-step framework for estimating univariate, unconditional {\it pdf}, and 
start by introducing some notations. 
Let $\mathscr{F}$ be the set of all strictly positive, univariate probability density functions on $[0,1]$.
Let $p_0 \in \mathscr{F}$ denote the underlying true density and 
$x_i \sim p_0$, $i=1,2,\dots, n$ be independent samples from $p_0$. 
Furthermore, let $\mathscr{F}_p$  be a pre-determined subset of $\mathscr{F}$, 
such that an optimal element (based on likelihood or any other desired criterion)) $f_{p} \in \mathscr{F}_p$  is relatively easy to compute. 
For instance, any parametric family with a simple maximum-likelihood estimator is a good candidate for $f_{p}$. 
Similarly, kernel density estimates are also good since 
they are computationally efficient and robust in univariate setups.

\begin{figure}
\begin{center}
\begin{tabular}{|c|c|}
\hline
\includegraphics[height=2.0in]{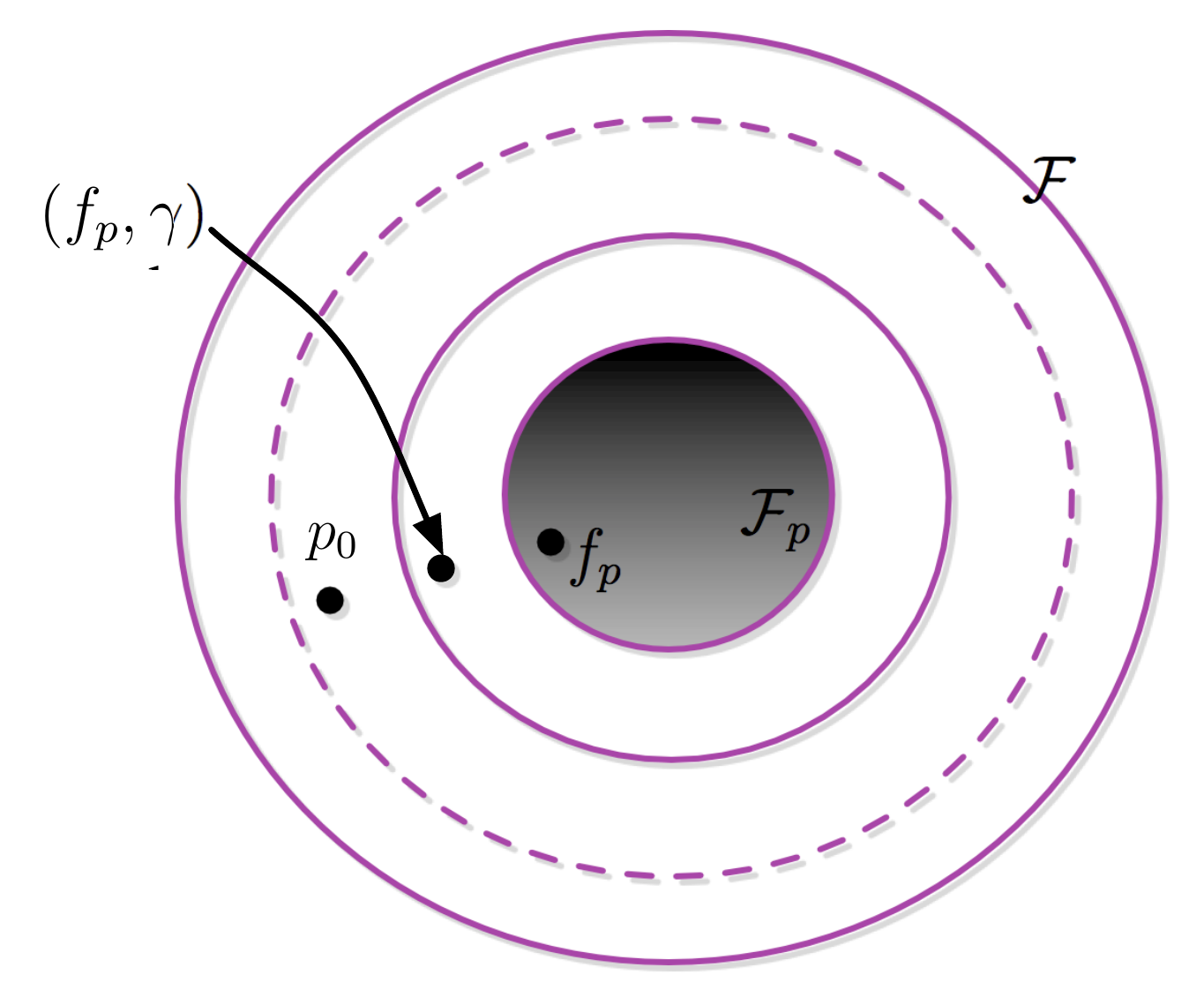} &
\includegraphics[height=2.0in]{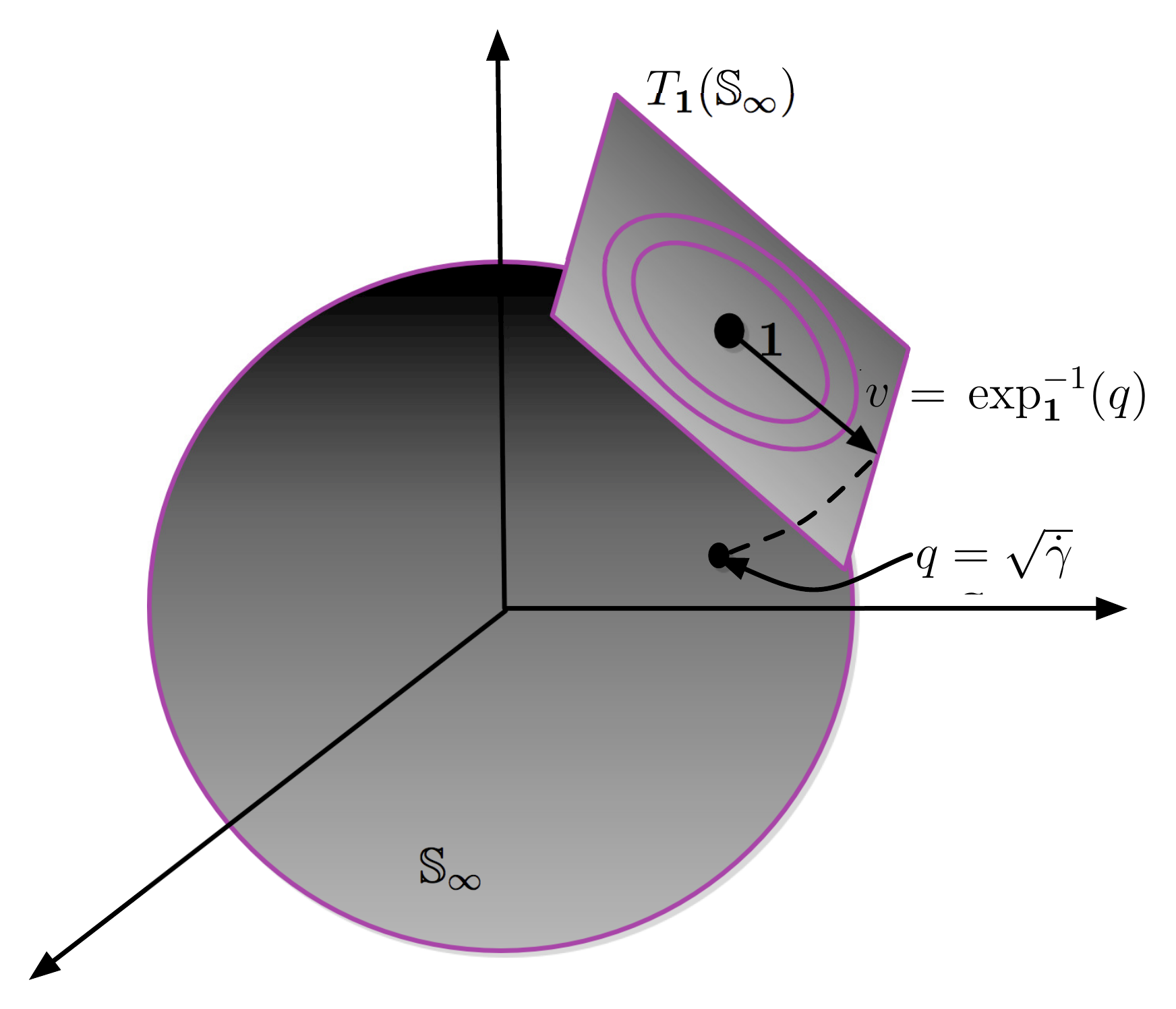}\\
\hline
\end{tabular}
\caption{\it Left: The true pdf $p_0$ is estimated by transforming an initial estimate $f_p$ by 
the warping function $\gamma$. The larger the set of allowed $\gamma$s, the better the estimate is. 
Right: Representing warping function $\gamma$ as element of the tangent space $T_{\bf 1}(\s_{\infty}^+)$.}
\label{fig:cartoon}
\end{center}
\end{figure}

Next, we define a warping-based transformation of elements of $\mathscr{F}_p$, using elements of $\Gamma$ defined earlier. 
Note that $\Gamma$ is an infinite-dimensional manifold that 
has a group structure under composition as the group operation. That is, for any 
$\gamma_1, \gamma_2 \in \Gamma$, the composition $\gamma_1 \circ \gamma_2 \in \Gamma$. 
The identity element of $\Gamma$ is given by $\gamma_{\mathrm{id}}(t) = t$, and for every $\gamma \in \Gamma$, 
there is a function $\gamma^{-1} \in \Gamma$ such that $\gamma \circ \gamma^{-1} = \gamma_{\mathrm{id}}$. 
For any $f_{p} \in \mathscr{F}_p$ and $\gamma \in \Gamma$, define the mapping 
$(f_{p}, \gamma) = (f_{p} \circ \gamma) \dot{\gamma}$ as  
given earlier. The importance of this mapping comes from the following result. 
\begin{prop}
The mapping $\mathscr{F} \times \Gamma \to \mathscr{F}$, specified above,  forms an action of $\Gamma$ on $\mathscr{F}$.
Furthermore, this action is transitive.  In other words, one can reach any element of $\mathscr{F}$, 
from any other element of $\mathscr{F}$ using an appropriate element of $\Gamma$.
\end{prop}
\noindent {\bf Proof}: We can verify the two properties in the definition of a group action: 
(1) For any $\gamma_1, \gamma_2 \in \Gamma$ and $f \in \mathscr{F}$, we have 
$((f, \gamma_1), \gamma_2) = (((f \circ \gamma_1) \dot{\gamma}_1) \circ \gamma_2) \dot{\gamma}_2 = 
(f, \gamma_1 \circ \gamma_2)$. (2) For any $f \in \mathscr{F}$, $(f, \gamma_{\mathrm{id}}) = f$. 
To show transitivity, we need to show that given any $f_1, f_2 \in \mathscr{F}$, there exists a $\gamma \in \Gamma$, 
such that $(f_1, \gamma) = f_2$. If $F_1$ and $F_2$ denote the cumulative distribution functions associated with 
$f_1$ and $f_2$, respectively, then the desired $\gamma$ is simply $F_{1}^{-1} \circ F_2$. Since $f_1$ is strictly positive, 
$F_1^{-1}$ is well defined and $\gamma$ is uniquely specified. Furthermore, since $f_2$ is strictly positive, we have $\dot{\gamma}>0$
and $\gamma \in \Gamma$. $\Box$

This result implies that together the pair $(f_{p}, \gamma)$ spans the full set $\mathscr{F}$, if $\gamma$ is 
chosen freely from $\Gamma$.  However, if one uses a proper
submanifold of $\Gamma$, instead of the full $\Gamma$, we may not reach the desired $p_0$ but only approximate it in some way. 
This intuition is depicted pictorially in the left panel of Figure \ref{fig:cartoon} where the inner disk denotes the set $\mathscr{F}_p$. The increasing 
rings around $\mathscr{F}_p$ represent the set $\{ (f_{p}, \gamma) | f_{p} \in \mathscr{F}_p\}$ with $\gamma$ belonging to progressively larger dimensional submanifolds of $\Gamma$. As the submanifolds approach the full space $\Gamma$, the corresponding approximation approaches $p_0$. The submanifolds are introduced formally in the next subsection. 
More details are also included in Section 6.1( Supplementary Materials).

\subsection{Finite-Dimensional Representation of Warping Functions} 
Given an initial estimate, the focus now shifts to the search for an optimal $\gamma$ such that the
warped density $(f_p \circ \gamma) \dot{\gamma}$ becomes the final estimate under the chosen criterion. 
However, solving sn optimization over $\Gamma$ faces two main challenges. 
First, $\Gamma$ is a nonlinear manifold, and second, it is infinite-dimensional. We handle the nonlinearity by forming 
a bijective map from $\Gamma$ to a tangent space of the unit Hilbert sphere $\s_{\infty}$ (the tangent space is a vector space), 
and infinite dimensionality by 
selecting a finite-dimensional subspace of this tangent space. 
Together, these two steps are equivalent to finding a  family of finite-dimensional submanifolds of $\Gamma$ that can be  {\it flattened} into vector spaces. 
This allows for a representation of $\gamma$ using elements of a Euclidean vector space and an application of standard optimization procedures. 

To locally flatten $\Gamma$, we define a function $q: [0,1] \to \real$, $q(t) = \sqrt{\dot{\gamma}(t)}$, termed the {\it square-root slope function} (SRSF)
of $\gamma \in \Gamma$. (For a discussion on SRSFs of general functions, please refer to Chapter 4 of 
\citet{srivastava2016functional}). 
For any $\gamma \in \Gamma$, its SRSF $q$ is an element of the interior of the positive orthant of  the unit Hilbert sphere $\s_{\infty} \subset \ltwo$, 
denoted by  $\s_{\infty}^+$. 
This is because
$\|q\|^2 = \int_0^1 q(t)^2 dt = \int_0^1 \dot{\gamma}(t) dt = \gamma(1) - \gamma(0) = 1$.
We have a positive orthant, boundaries excluded, because by definition $q$ is a strictly positive function. The mapping between $\Gamma$ and 
$\s_{\infty}^+$ is a bijection, with its inverse given by $\gamma(t) = \int_0^t q(s)^2 ds$. 
The set $\s_{\infty}$ is a smooth manifold with known geometry under the $\ltwo$ Riemannian metric 
\citet{lang2012fundamentals}. Although is not a vector space,  it can be easily flattened 
into a vector space (locally) due to its constant curvature. A natural choice for flattening is the 
vector space tangent to $\s_{\infty}^+$ at the point ${\bf 1}$, 
which a constant function with value $1$. (${\bf 1}$ is the SRSF corresponding to $\gamma=\gamma_{\mathrm{id}}(t) = t$.)
The tangent space of $\s_{\infty}^{+}$ at ${\bf 1}$
is an infinite-dimensional vector space given by: 
$T_{{\bf 1}}(\s_{\infty}^+) = \{ v \in \ltwo([0,1],\real) | \int_0^1 v(t) dt = \inner{v}{{\bf 1}}= 0\}$. See the right panel of Fig. 
\ref{fig:cartoon} for an illustration of this idea. 
Next, we define a mapping that takes an arbitrary element of $\s_{\infty}^+$ to this tangent space. For this {\it retraction}, 
we will use the inverse
exponential map; it takes $q \in \s_{\infty}^+$ to $T_{{\bf 1}}(\s_{\infty}^+)$  according to:  
\begin{equation}
\exp^{-1}_{{\bf 1}} (q) :\s_{\infty}^+ \to T_{\bf 1}(\s_{\infty}^+),\ \ \ 
v= \exp^{-1}_{{\bf 1}} (q)= {\theta \over \sin(\theta)} (q- {\bf 1} \cos(\theta ))\ ,
 \end{equation} 
 where $\theta =\cos^{-1}(\inner{{\bf 1}}{q})$ is the arc-length from $q$ to ${\bf 1}$. 
 The right panel of Fig. \ref{fig:cartoon} also shows the mapping from $\s_{\infty}^+$
 to $T_{{\bf 1}}(\s_{\infty}^+)$.

We impose a natural Hilbert structure on 
$T_{{\bf 1}}(\s_{\infty}^+)$ using
 the standard inner product: $\inner{v_1}{v_2} = \int_0^1 v_1(t)v_2(t) dt$. It is easy to check that since $q \in \s_{\infty}^+, \theta =\cos^{-1}(\inner{{\bf 1}}{q})<\pi/4$, 
 and hence $\|v\|=\sqrt{ \int_0^1 v(t)^2 dt}=\theta <\pi/4$, where $v= \exp^{-1}_{{\bf 1}} (q)$. 
Thus, the range of the inverse exponential map is not the entire $T_{{\bf 1}}(\s_{\infty}^+)$, but an 
open subset $T_{{\bf 1}}^0(\s_{\infty}^+)=\{ v \in T_{{\bf 1}}(\s_{\infty}^+): \|v\| <\pi/4\}$.
Further, we can select any orthogonal basis  ${\cal B} = \{ b_j, j=1,2,\dots\}$ of the Hilbert space $T_{{\bf 1}}(\s_{\infty}^+) $ to express
its elements $v$ by their corresponding coefficients; 
that is,  $v(t) = \sum_{j=1}^{\infty} c_j b_j(t)$, where $c_j= \inner{v}{b_j}$. The only restriction on the basis elements $b_j$'s is that they must be orthogonal to {\bf 1}, that is, $\inner{b_j}{{\bf 1}}=0$. 
In order to map points back from the tangent space to the Hilbert sphere, we use the exponential map, 
given by: 
\begin{equation}
\exp (v) :T_{\bf 1}(\s_{\infty}^+) \to\s_{\infty} ,\ \ \ 
 \exp(v)= \cos(\|v\|){\bf 1} + {\sin(\|v\|) \over \|v\|}\ .
 \end{equation} 
 If we restrict the domain of the exponential map to 
the subset $T_{\bf 1}^0(\s_{\infty}^+)$, then the range of this map is $\s_{\infty}^+$.
Using these two steps, we specify the finite-dimensional, therefore approximate,
representation of warpings. We define a composite map $H: \Gamma \to \real^J$, illustrated in Figure \ref{fig:cartoon3}, as
\begin{equation}
\gamma  \in \Gamma ~~~~~\xrightarrow{\mbox{SRSF}} ~~~~~q = \sqrt{\dot{\gamma}}  \in \s_{\infty}^+ ~~~~~~~
\xrightarrow{\exp^{-1}_{\bf 1}} ~~~~
v 
\in T_{{\bf 1}}^0(\s_{\infty}^+)  ~~~~ \xrightarrow{ \{b_j\}} ~~~~ 
\{c_j = \inner{v}{b_j} \} \in \real^J \ .
\label{eq:representation}
\end{equation}
The range of $H$ is $V_{\pi}^J = \{c \in \real^J : \| \sum_{j=1}^{J} c_j b_j\| < \pi/4\} \subset \real^J$. Now, we define $G:\real^J \to \Gamma$, as
\begin{align}
\{c_j\} \in \real^J \xrightarrow{\{b_j\}} ~ v = \sum_{j=1}^J c_j b_j  \in T_{{\bf 1}}(\s_{\infty}^+)~ \xrightarrow{\exp_{\bf 1}} ~q = 
\exp_{\bf 1}(v) ~\xrightarrow{~}~   \gamma(t) = \int_0^t q(s)^2 ds\ .
\end{align}
If we restrict the domain of $G$ to $V_{\pi}^J$, then $G$ is invertible and its inverse is $H$. 
Restricting our focus to only the set $V_{\pi}^J$, rather than the entire space $\real^J$, 
we identify the function $G$ as $H^{-1}$.
For any $c \in V_{\pi}^J$, let $\gamma_c$ denote the diffeomorphism $H^{-1}(c)$.
For any fixed $J$, the set $H^{-1}(V_{\pi}^J)$ is a $J$-dimensional submanifold of $\Gamma$,and we pose the estimation 
problem on this submanifold. 
As $J$ goes to infinity, this submanifold converges to the full group $\Gamma$.


\begin{figure}
\begin{center}
\includegraphics[height=2.78in]{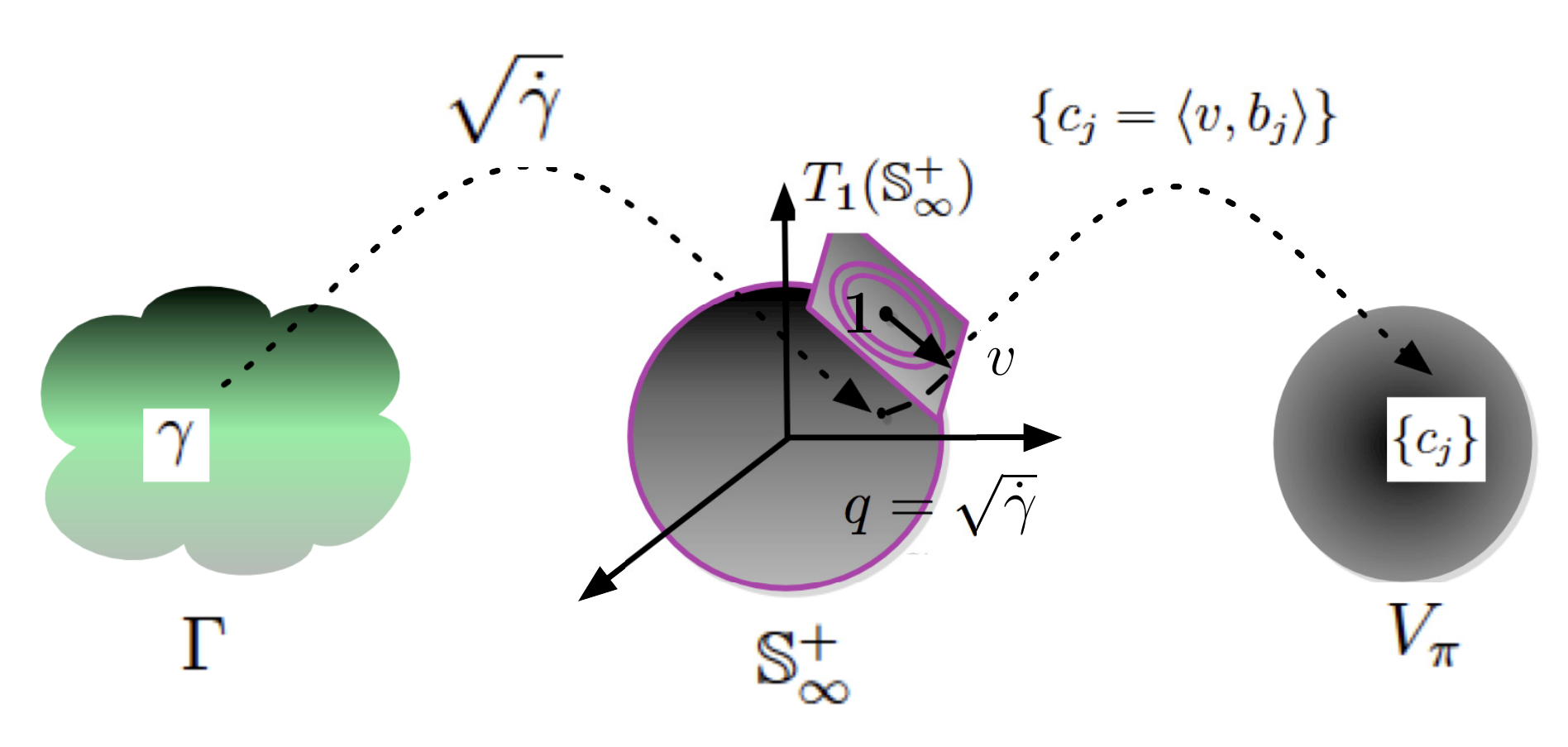}
\caption{A graphic representation of Eqn. \ref{eq:representation} leading to a bijective map between 
$\Gamma$ and $V_{\pi}^J$.}
\label{fig:cartoon3}
\end{center}
\end{figure}

With this setting, we can rewrite the estimation of the unknown density $p_0$, given an initial estimate $f_p$, as
$\hat{f}(t)= f_p(\gamma_{\hat{c}} (t))\dot{\gamma}_{\hat{c}}(t), t\in [0,1]$, where $\gamma_{\hat{c}} = H^{-1}(\hat{c})$ and 
\begin{equation}
\hat{c} =  \argmax_{c \in V_{\pi}^J} \left(  \sum_{i=1}^n \Bigg[ \log \left(f_p \left(\gamma_c(x_i)\right)\dot{\gamma}_c(x_i)\right) \Bigg] \right)\ .
\label{eq:opt}
\end{equation}
The truncated basis approximation takes place in the tangent space representation of 
$\Gamma$, rather than in the original space as is the case in \citet{birge1998minimum}, \citet{donoho1996density} and several others. 
The tangent space approximation is superior because it is a flat space whereas $\Gamma$ or $\s_{\infty}^+$ are not flat. \\

\noindent {\bf Choice of Basis Functions}: 
Now that we are in a Hilbert space $T_{\bf 1}(\s_{\infty})$, we
 can choose from a wide range of basis elements. For example, 
one can use the Fourier basis elements (excluding ${\bf 1}$ of course).
However, other bases such as splines and Legendre polynomials can also be used. 
In the experimental studies,  we demonstrate an example using the Meyer wavelets that have  
attractive properties of infinite differentiability and support over all reals. \citet{vermehren2015close} provides 
a closed-form expression for Meyer wavelets and scale function in the time domain, 
which enables us to use the basis set for representation. 
However, Meyer wavelets are not naturally orthogonal to ${\bf 1}$ and so they need to be orthogonalized first but that
can be done offline.

\subsection{Advantages Over Direct Approximations}

In the previous section, we have used the geometry of $\Gamma$ to develop a natural, local flattening of $\Gamma$. 
Other, seemingly simpler, choices are also possible but at some cost in estimation performance. For instance,
since any $\gamma$ can also be viewed as a nonnegative function in $\ltwo$ with appropriate 
constraints, it may be tempting to use $\gamma(t) = \sum_{j=1}^{\infty} c_j b_j(t)$, for some orthogonal basis  ${\cal B} = \{ b_j, j=1,2,\dots\}$ of $\ltwo[0,1]$ as in \citet{hothorn2015most}. This seems easier than our approach as it avoids going through a nonlinear transformations. 
However, the fundamental issue with such an approach is that 
$\Gamma$ is a nonlinear manifold and one cannot technically express and estimate elements of $\Gamma$
directly using linear representations. \citet{hothorn2015most} uses Bernstein polynomials, with monotonically increasing coefficients, 
to represent elements of $\Gamma$. However, one does not reach 
the entire set $\Gamma$ using such a representation. To be specific, it is easy to find a significant 
subset of $\Gamma$ whose elements cannot be represented in this system. 
As a simple example, consider a 
$\gamma= \sum_{i=0}^{4} c_i B_{i,4}$ with $c_0=0$, $c_1=0.4$, $c_2=0.3$,$c_3 = 0.5$, $c_4 = 1$ (not satisfying the monotonicity constraint). 
Here, $B_{i,4}$ refer to the Bernstein basis elements of order $4$.
Even though this $\gamma$ is a proper diffeomorphism, it  cannot be represented in the system used by \citet{hothorn2015most}.

Another issue in directly approximating element of $\Gamma$ that
 both $\gamma$ and $\dot{\gamma}$ are present the final estimate and one needs a good approximation of both of these functions. 
 However, a good approximation of $\gamma$ does not imply a good approximation of $\dot{\gamma}$. In contrast, the reverse holds true
 as shown next. 
\begin{prop}
For any  $\gamma \in \Gamma$, let $\dot{\gamma}_{\mathrm{app}}$ be an approximation of $\dot{\gamma}$,
 and let $\gamma_{\mathrm{app}}$ be the integral of $\dot{\gamma}_{\mathrm{app}}$. For all $x_0 \in (0,1]$ consider intervals $I_{x_0}$ of the form $[0,x_0]$. Then, on all intervals $I_{x_0}$, ${\|\gamma - \gamma_{\mathrm{app}}\|}_{\infty} \leq {\|\dot{\gamma} - \dot{\gamma}_{\mathrm{app}}\|}_{\infty}$.
\end{prop}
 
 \noindent {\bf Proof}: 
 Let $t \in I_{x_0}. |\gamma(t) -\gamma_{\mathrm{app}}(t)| = |\int_0^t \dot{\gamma}(s) ds -\int_0^t \dot{\gamma}_{\mathrm{app}}(s) ds| \leq \int_0^t |\dot{\gamma}(s) - \dot{\gamma}_{\mathrm{app}}(s)| ds \leq {\|\dot{\gamma} - \dot{\gamma}_{\mathrm{app}}\|}_{\infty}.t \leq {\|\dot{\gamma} - \dot{\gamma}_{\mathrm{app}}\|}_{\infty}.x_0 \leq {\|\dot{\gamma} - \dot{\gamma}_{\mathrm{app}}\|}_{\infty} $
   $\Box$
This proposition states that a good approximation of $\dot{\gamma}$
ensures a good approximation of $\gamma$, and supports our approach of approximating $\gamma$ 
via the inverse exponential transformation of its SRSF to the tangent space $T_{{\bf 1}}(\s_{\infty}^+)$. 
On the other hand,  a direct approximation of $\gamma$ will needs many more basis elements to ensure a good approximation of $\dot{\gamma}$.

\subsection{Estimation of Densities with Unknown Support}
 \label{boundary}
So far we have restricted to the 
interval $[0,1]$ for representing a {\it pdf}. However, the framework extends naturally to {\it pdf}s with unknown support. 
For that, we simply scale the observations to $[0,1]$ and carry out the original procedure.
Let $X_1,X_2, \dots , X_n \sim p_0$, where $X_i$s are $n$ independent observations from a density $p_0$ 
with an unknown support. We transform the data as
$Y_i=\frac{X_i - A}{B-A}$, 
where $A$ and $B$ are the estimated boundaries of the density. Following \citet{turnbull2014unimodal}, we take $A=X_{(1)}-s_{X}/\sqrt{n}$, and  $A=X_{(n)}+s_{X}/\sqrt{n}$,
 where $X_{(1)}$ and $X_{(n)}$ are the first and last order statistics of X, and $s_{X}$ is the sample standard deviation of the observed samples. Using the scaled data, we can find the estimated {\it pdf} $f_w$ on $[0,1]$ and then undo the scaling to reach the
final solution.
\citet{turnbull2014unimodal} provide a justification for the choice of $A$ and $B$ as the estimates for the bounds of the density. They also discuss
an alternate way of estimating the boundaries using ideas presented in \citet{de2011confidence}, and suggest that the Carvalho method produces wider and more conservative boundary estimates. 

Finally, using the fact that any piecewise continuous density function, with support $\mathbb{R}$ and range $\mathbb{R}_{\ge 0}$ , can be approximated
to any desired degree by a strictly positive density function on some bounded interval $[A,B]$ (under $\ltwo$ norm, for example) , we can extend our method to this 
larger class of functions.





\setcounter{equation}{0} 
\section{Asymptotic Analysis and Convergence Rate Bounds}
We have represented an arbitrary {\it pdf} as a function of the coefficients {\it w.r.t} a basis set of the tangent space. 
We note that in order to represent the entire space $\mathscr{F}$, we need a Hilbert basis with infinitely many elements. 
However, in practice, we use only a finite number $J$ of basis elements. Hence, we are actually optimizing over a subset of the space of density functions based on only a few basis elements and using it to approximate the true density. This subset is called the 
{\it approximating space}. Since we are performing maximum likelihood estimation over an approximating space for {\it pdf}s, 
our estimation is  akin to the sieve MLE, discussed in \citet{wong1995probability}.

First, we introduce some notations. 
Recall that $\mathscr{F}$ is  the space of all univariate, strictly positive {\it pdf}s  on $[0,1]$ and zero elsewhere. 
Let $\mathscr{F}_n$ be the approximating space of $\mathscr{F}$ when using $J=k_n$ basis elements for the tangent space
$T_{{\bf 1}}(\s_{\infty}^+)$, where $k_n$ is some function of the number of observations $n$.
\label{fnnotation}
Let $f_p \in \mathscr{F}_p \subset \mathscr{F}$ be an initial estimate, and let
$\mathscr{F}_n=\{f_p (\gamma) \dot{\gamma},  \gamma = H^{-1}(c))|\ \ c \in V_{\pi}^J \subset {\mathbb{R}}^{k_n}\}$, where $H$ and $V_{\pi}^J$ are defined in Section 2.1.   
As $n \rightarrow \infty, k_n \rightarrow \infty$. So $\mathscr{F}_n \rightarrow \mathscr{F}$ as $n \rightarrow \infty$.
Let $\eta_n$ be a sequence of positive numbers converging to 0. Let $\mathscr{Y}^{(n)}$ be the space of $n$ observed points. We call an estimator $\hat{p}:\mathscr{Y}^{(n)} \rightarrow  \mathscr{F}_n$ an $\eta_n$ sieve MLE if
\[
\frac{1}{n}\sum_{i=1}^{n} \log \hat{p}(Y_i) \geq  \underset{p \in \mathscr{F}_n}{\text{sup}} \frac{1}{n}\sum_{i=1}^{n}\log p(Y_i) -\eta_n
\]
In the proposed method, the estimated {\it pdf}  is exactly $ \underset{p \in \mathscr{F}_n}{\text{sup}} \frac{1}{n}\sum_{i=1}^{n}\log p(Y_i)$. 
Therefore, this estimate is a sieve MLE with $\eta_n \equiv 0$. 
Let $p_0$  denote the true density which is assumed to belong a H\"{o}lder space of order $\beta>0$.   By the equivalence of the pdf space and the coefficient space of expansion of $\gamma$ (refer to Appendix S1.1), it is straightforward to show that if $k_n=l_1n^{1/(2\beta +1)}$ then $\inf_{f \in \mathcal{F}_n} \| p_0 - f\|_\infty \leq l_2 n^{-\beta/(2\beta+1)}$ for some arbitrary constants $l_1$ and $l_2$ .  This follows from standard approximation results in $\ltwo$ basis (e.g. Fourier) of  H\"{o}lder functions of order $\beta$. For a detailed discussion please refer to \citet{triebel2006theory}. 

To control the approximation error, \citet{wong1995probability} introduces a family of discrepancies. They define $\delta_n (p_0,\mathscr{F}_n)=\text{inf}_{f \in \mathscr{F}_n} \rho (p_0,f)$, called the $\rho$-approximation error at $p_0$. 
The control of the approximation error of $\mathscr{F}_n$ at $p_0$ is necessary for obtaining results on the convergence rate for sieve MLEs.
We follow \citet{wong1995probability} to introduce a family of indexes of discrepency in order to formulate the condition on the approximation error of $\mathscr{F}_n$.
Let
\[
g_{\alpha} (x) = \left\{\begin{array}{lr}(1/\alpha)[x^{\alpha} -1], -1<\alpha<0 \text{ or }0<\alpha \leq 1\\
\log{x}, \text{  if }\alpha=0+
\end{array}
\right.
\]
Set $x=p/f$ and define $\rho_{\alpha} (p,f)=E_pg_{\alpha}(X)=\int pg_{\alpha}(p/f).$
We define $\delta_n (\alpha) =\inf _{f \in \mathscr{F}_n} \rho_{\alpha} (p_0,f)$. We use $\alpha=1$ for our results. Then $\delta_n(1)=\int {(p_0 -f)}^2/f $.

The $\delta$-cover of a set $T$ wrt a metric $\rho$ is a set $\{\Theta^{1},\dots,\Theta^{N}\}\subset T$ such that for each $\Theta \in T$, there exists some $i\in \{1,\dots,N\}$ with $\rho(\Theta,\Theta_i)\leq \delta$. The covering number $N$ is the cardinality of the smallest delta cover. Then $\log(N)$ is the metric entropy for $T$. The following Lemma provides a bound for the Hellinger metric entropy for $\mathscr{F}_n$.

\begin{lem}
There exists positive constants $C_3$ and $C_4$ and a positive $\epsilon<1$ such that, 
\begin{equation}
\int_{{\epsilon}^2/2^8}^{\sqrt{2}\epsilon} {H}^{1/2} (\frac{u}{C_3},\mathscr{F}_n) du \leq C_4n^{1/2}{\epsilon}^2,
\end{equation}
\label{hellinger}
\end{lem}
The following corollary provides a uniform exponential bound on likelihood ratio surfaces and follows from Lemma \ref{hellinger} due to Theorem $1$ of\cite{wong1995probability}.
\begin{corollary}\label{thm:1}
If Lemma \ref{hellinger} holds, there exists positive constants $C_1$ and $C_2$ such that for any $\epsilon>0$, 
\[
P^{*}\bigg( \underset{\{{\|p^{1/2}-p_{0}^{1/2}\|}_2 \geq \epsilon , p\in \mathscr{F}_n\}}{\text{sup}} \prod_{i=1}^{n} p(Y_i)/p_0(Y_i) \geq \text{exp}(-C_1n{\epsilon}^2)\bigg)\leq 4 \text{ exp}(-C_2n{\epsilon}^2)
\]

\end{corollary}

\begin{lem}
There exists a positive constant $C_5$ such that $\delta_n(1) = C_5n^{-2\beta/(2\beta+1)}$.
\end{lem}
The following theorem  provides convergence rates of the sieve estimators.   
\begin{theorem}\label{thm:2} Under the assumptions listed  above,
let $C_1, \dots, C_4$, be as in Lemma \ref{hellinger} and Corollary \ref{thm:1}.  Define,
$\epsilon_n^*=Mn^{-\beta/(2\beta +1)}\sqrt{\log n}$  for some $M>0$.
Then if $\delta_n(1)<1$,
\begin{eqnarray}\label{eq:main}
 P({\|q^{1/2}-p_{0}^{1/2}\|}_2 \geq \epsilon_n^* ) \leq 5\text{exp}\big(-C_2n{(\epsilon_n^* )}^2\big) + \text{exp}\big(-\frac{1}{4}n\alpha C_1{(\epsilon_n^* )}^2\big).  
\end{eqnarray}
 
\end{theorem}
\noindent The proofs of the results are deferred to Section 6 (Supplementary Materials). Note that the convergence rate is independent of the initial step $f_p$ (upto constant terms) because the estimation problem is shifted to $\Gamma$ given a fixed choice of $f_p$.

\setcounter{equation}{0} 

\setcounter{equation}{0} 
\section{Simulation Studies}
Next, we present results from experiments on univariate unconditional density estimation procedure 
involving two simulated datasets. 
The computations described here
are performed on an Intel(R) Core(TM) i7-3610QM CPU processor laptop, and the 
computational times are reported for each experiment. 
We compare the proposed solution with two standard techniques: (1) 
 kernel density estimates with bandwidth selected by unbiased cross validation method, henceforth referred to as {\it kernel(ucv)}, 
 (2) a  standard Bayesian technique using the  function {\it DPdensity} in the R package \texttt{DPPackage}. 
 We focus on the average performance of the different techniques over $100$ independent 
samples from the true density. We use {\it ksdensity} as 
the initial estimate $f_p$  for our approach. We consider sample sizes of $25, 100$ and $1000$, to 
study the effect of $n$ on estimation performance and computational cost. 
The performance is evaluated using multiple norms:  $\ltwo$, $\lone$ norm and $\linf$ norm, averaged over
the $100$ samples. 

We borrow the first example from \citet{tokdar2007towards} and \citet{lenk1991towards},
where $p_0 \propto 0.75 \text{exp}(\text{rate}=3) + 0.25 {\cal N}(0.75,2^2)$, 
a mixture of exponential and normal density truncated to the interval $[0,1]$: 
Table \ref{tabletokdar} summarizes estimation performance and computation cost for these methods at different sample sizes. The values of mean and standard
deviation have been scaled by $100$ for convenience. It is observed
that when $n=25$, {\it kernel(ucv)} method outperforms the other two methods. However, for higher sample sizes, the 
warping-based method has a better overall performance. The computational cost of the 
proposed method, while higher than {\it kernel(ucv)}, is much less 
than the {\it DPdensity} for higher sample sizes. In this example, we also studied performance using the Fourier basis and 
the results were very similar.  

\begin{table}[t!] 
\caption{\it A comparison of the performances for mixture of exponential and normal example.}
\label{tabletokdar}\par
\vskip .2cm
\centerline{\tabcolsep=3truept\begin{tabular}{|rcrrrrrrrrr|} \hline 
\multicolumn{2}{|c}{Method:}
         & \multicolumn{3}{c}{DPDensity}
         & \multicolumn{3}{c}{Kernel(ucv)}
         & \multicolumn{3}{c|}{Warped Estimate} \\ \hline
\multicolumn{1}{|c}{$n$} &  \multicolumn{1}{c}{Norm}
          & \multicolumn{1}{c}{Mean} & \multicolumn{1}{c}{std.dev.} & \multicolumn{1}{c}{Time}
          & \multicolumn{1}{c}{Mean} & \multicolumn{1}{c}{std.dev} & \multicolumn{1}{c}{Time}
          & \multicolumn{1}{c}{Mean} & \multicolumn{1}{c}{std.dev} & \multicolumn{1}{c|}{Time} \\[3pt] \hline
\multirow{ 3}{*}{25} & $\lone$ &37.26 & 8.63 &  & 33.51 & 11.97 &  & 39.53 & 9.8 &   \\
& $\ltwo$ & 5.05 & 0.9 & 4 sec & 4.5 & 1.44 & $<1$ sec & 4.96 & 1.27 & 5 sec  \\
& $\linf$ & 1.64 & 0.21 &  & 1.44 & 0.47 &  & 1.34 & 0.53 &
 \\ \hline
\multirow{ 3}{*}{100} & $\lone$ & 22.87 & 5.32 &  & 21.9 & 5.54 &  & 22.46 & 4.95 &    \\
& $\ltwo$ & 3.47 & 0.58 & 18 sec & 3.14 & 0.57 & $<1$ sec & 2.93 & 0.61 & 5 sec  \\
& $\linf$ & 1.49 & 0.2 &  & 1.23 & 0.24 &  & 0.88 & 0.34 &
   \\ \hline
\multirow{ 3}{*}{1000} & $\lone$ & 10.79 & 2.05 &  & 11.57 & 2.14 &  & 10.05 & 1.36 &   \\
& $\ltwo$ & 1.83 & 0.24 & 225 sec & 1.67 & 0.23 & $<1$ sec & 1.31 & 0.16 & 5 sec  \\
& $\linf$ & 1.18 & 0.2 &  & 0.88 & 0.22 &  & 0.5 & 0.17 &  
\\ \hline
\end{tabular}}
\end{table}

For the second example we take Example 10 from \citet{marron1992exact}, which uses 
a claw density: $p_0  = \frac{1}{2}{\cal N}(0,1) + \sum_{l=0}^{4} \frac{1}{10}{\cal N}(\frac{l}{2} -1, {(0.1)}^2)$. 

Unlike the previous example, instead of fixing $J$, the number of tangent basis elements, 
we employ Algorithm 1 (please refer to Section 7 of the Supplementary Materials)  to find the optimal $J$ based on the AIC, with a maximum allowed value of $40$ basis elements. 
Consequently, as can be seen in Table \ref{tableclaw},  the computation cost goes up.
Additionally, we note that the cost is  highest for $n=25$ and actually decreases as $n$ increases.  
This is because for small $n$ there is less information and it take more time for the objective function to converge.  

\begin{table}[t!] 
\caption{\it Comparison for claw density example.}
\label{tableclaw}\par
\vskip .2cm
\centerline{\tabcolsep=3truept\begin{tabular}{|ccrrrrrrrrr|} \hline 
\multicolumn{2}{|c}{Method:}
         & \multicolumn{3}{c}{DPDensity}
         & \multicolumn{3}{c}{Kernel(ucv)}
         & \multicolumn{3}{c|}{Warped Estimate} \\ \hline
\multicolumn{1}{|c}{$n$} &  \multicolumn{1}{c}{Norm}
          & \multicolumn{1}{c}{Mean} & \multicolumn{1}{c}{std.dev.} & \multicolumn{1}{c}{Time}
          & \multicolumn{1}{c}{Mean} & \multicolumn{1}{c}{std.dev} & \multicolumn{1}{c}{Time}
          & \multicolumn{1}{c}{Mean} & \multicolumn{1}{c}{std.dev} & \multicolumn{1}{c|}{Time} \\[3pt] \hline
\multirow{ 3}{*}{25} & $\lone$ &39.15 & 6.29 &  &17.06 & 2.33 &  & 18.28 & 3.3 &   \\
& $\ltwo$ & 5.46 & 0.48 & 4 sec & 2.09 & 0.3 & 1 sec & 2.41 & 0.43 & 105 sec   \\
& $\linf$ & 1.2 & 0.05 &  & 0.5 & 0.14 &  & 0.64 & 0.17 &  
 \\ \hline
\multirow{ 3}{*}{100} & $\lone$ & 28.39 & 4.55 &  & 8.54 & 2.38 &  & 9.06 & 2.6 &   \\
& $\ltwo$ & 4.31 & 0.46 & 26 sec & 1.18 & 0.28 & 1 sec & 1.3 & 0.35 & 85 sec  \\
& $\linf$ & 1.08 & 0.09 &  & 0.34 & 0.08 &  & 0.42 & 0.13 &  
 \\ \hline
\multirow{ 3}{*}{1000} & $\lone$ & 19.28 & 1.63 &  & 2.4 & 0.38 &  & 2.46 & 0.43 &   \\
& $\ltwo$ & 3.16 & 0.15 & 331 sec & 0.38 & 0.06 & 1 sec & 0.4 & 0.08 & 71 sec  \\
& $\linf$ & 0.83 & 0.04 &  & 0.14 & 0.03 &  & 0.15 & 0.04 &  
 \\ \hline
\end{tabular}}
\end{table}

Table \ref{tableclaw} shows that at  $n=1000$, the performances of 
all three methods are similar, especially between {\it kernel(ucv)} and warped density estimate. 
In fact, the warped density estimate and {\it kernel(ucv)} perform similarly 
even at low sample sizes, while {\it DPdensity} performs poorly. 
These results were obtained using the Fourier basis but the results for Meyer basis were similar.


\setcounter{equation}{0} 
\section{Extension to Conditional Density Estimation}
The idea of using diffeomorphisms to warp an initial density estimate, while maximizing likelihood, extends naturally to conditional density estimation.  
Consider the following setup: Let $X$ be a fixed $d$-dimensional random variable with a positive density on its support. 
Let $Y \sim  p_0(m(X),\sigma_X^2)$, where $p_0$ is the unknown conditional density that changes smoothly with $X$;
$m(X)$ is the unknown mean function, assumed to be differentiable; and, $\sigma_X^2$ is the unknown variance, 
which may or may not depend on $X$. $Y$ is assumed to have a univariate, continuous distribution with  support on 
unknown interval $[A,B]$. 
 We observe the pairs 
$(Y_i,X_i),i=1,\dots,n$, and are interested in recovering the conditional density $p_0(m(X),\sigma_{X}^2)$.

In order to initialize estimation,  we assume a nonparametric mean regression model of the form 
$y_i =  m(x_i) + \epsilon_i$ ,  $\epsilon_i \sim f_p(0,\sigma^2)$, where $m(\cdot)$ is estimated using  standard local linear regression, 
$f_p$ is an initial estimate for the conditional density of the response variable, 
and ${\sigma^2}$ is estimated using the sample standard deviation of the residuals $Y_i -\hat{m}(X_i)$. We have used truncated normal density as  $f_p$ in the experiments presented later but other choices are equally valid. As was the case in 
unconditional {\it pdf} estimation, it is not required that the initial 
estimate has mean function close to the true mean function, or assume any particular form. 
The only requirement is 
that the initial conditional density should be continuous and bounded away from zero, and the density should vary smoothly with $X$ in the sense that if $X_1$ and $X_2$ are close to each other, then $f_p(Y|X_1)$ should be close to $f_p(Y|X_2)$ in the $\ltwo$ or some other metric. 
Let $F_{p,x_0}$ be the corresponding initial estimate of the conditional distribution function of $Y$, 
given $X=x_0$ for some given value of the predictor $x_0$. 
Then, the warped density estimate, for a warping function $\gamma$ and location $x_{0}$, is 
$
f_{w,x_0} (y|X=x_0)= f_{p}(\gamma (y),\hat{m}(x_0),\hat{\sigma}^2)\dot{\gamma}(y)$.
If $F_{t,x_0}$ is the true conditional distribution function of $Y$,  given $X=x_0$, then the true $\gamma$ 
at location $x_0$ is $\gamma_{x_0} = F_{p,x_0}^{-1} \circ F_{t,x_0}$.
Setting $f_{p,x_0} \equiv f_p(\hat{m}(X),\hat{\sigma}^2)$, 
we estimate the optimal $\gamma$ by a weighted maximum likelihood estimation:
$
\hat{\gamma}_{x_0} = \argmax_{\gamma \in \Gamma} \left( \sum_{i=1}^n \log\bigg[(f_{p,x_0}(\gamma (y_i)|x_i)\dot{\gamma})W_{x_0,i}\bigg] \right)\ ,
$
where $W_{x_0,i}$ is the localized weight associated with the $i$th observation, calculated as:
$$
W_{x_0,i}=\frac{{\cal N}({\|X_i-x_0\|}_2/h(x_0);0,1)}{\sum_{j=1}^{n}{\cal N}({\|X_j-x_0\|}_2/h(x_0);0,1)}
$$
where ${\cal N}(\cdot;0,1)$ is the standard normal {\it pdf} and $h(x_0)$ is the parameter that controls the relative weights associated with the observations. However, weights defined in this way results in higher bias because information is being borrowed from all observations. As discussed in an example in \citet{bashtannyk2001bandwidth}, we allow only a specified fraction of the observations $X_i$ to have a positive weight. However, using too small a fraction will result in unstable estimates and poor practical performance because the effective sample size will be too small. Hence we advocate using the nearest $50\%$ of the observations (nearest to the target location) for borrowing information and then calculating the weights for this smaller sample as defined before.

The parameter $h(x_0)$  is akin to the bandwidth parameter associated with traditional kernel methods for density estimation. 
A very large value of $h(x_0)$ distributes approximately equal weight to all the observations, whereas a
very small value considers only the observations in a small neighborhood around $x_0$. 
Since $h(x_0)$ is scalar, the tremendous computational cost associated with obtaining cross-validated bandwidths in each predictor dimension, when the predictor dimension is high, is avoided. 
When the predictor is one-dimensional, the parameter $h(x_0)$ is chosen according 
to the location $x_0$ using a two-step procedure as follows:
\begin{enumerate}
\item
Compute a standard kernel density estimate $\hat{K}$ of the predictor space using a fixed bandwidth chosen according to any standard criterion. Let $h$ be the fixed bandwidth used.

\item
Then, set the bandwidth parameter $h( x_0)$ at location $x_0$ to be $h(x_0)=h/\sqrt{\hat{K}(x_0)}$.
 
\end{enumerate}
The intuition is that $h$ controls the overall smoothing of the predictor space based on the sample points, and the $\sqrt{\hat{K}(x_0)}$ stretches or shrinks the bandwidth at the particular location. 
 The choice of the adaptive bandwidth parameter is motivated from the variable bandwidth kernel density estimators discussed in \citet{terrell1992variable}, \citet{van2003adaptive} and \citet{abramson1982bandwidth}, among others. 
In case of $d$ independent predictors, $h({\bf x_0})$ at ${\mathbf x_0}$ is chosen as follows:
\begin{enumerate}
\item
Compute the kernel density estimate $\hat{K}_i,i \in 1,\cdots,d$ for the $d$ predictors with associated bandwidths $h_1,h_2,\cdots,h_d$. Then $h$ is chosen as the harmonic mean of the $h_i$'s.
\item
Once $h$ is obtained, the bandwidth parameter $h({\mathbf x_0})$ at ${\bf x_0}$ is given by:\\
\begin{equation}
h({\mathbf x_0})=h/\bigg(\prod_{i=1}^{d}\sqrt{\hat{K}_i(x_{0i})}\bigg)\label{multi}
\end{equation}
 where $x_{0i}$ is the $i$th coordinate of ${\mathbf x_0}$.
\end{enumerate}
This choice of using the harmonic mean is 
based on the dependence of the minimax rates of convergence of estimators to the harmonic mean of the smoothness of the density along the different dimensions, as discussed in \citet{lepski2015adaptive}.

\subsection{Simulation Studies}
We present two examples to illustrate the proposed method and compare it with a standard R package \texttt{NP} (with kd-tree package implementation to reduce computation time). 
In these experiments we have used a gaussian family for $f_p$, 
the initial parametric conditional density estimate. To estimate the mean function, we have used a local-linear 
regression function with gaussian kernel weights and bandwidth obtained from \texttt{kernel(bcv)} available in \texttt{R} package \texttt{kedd}. Bandwidth from other estimators like unbiased cross validation and even the naive {\it ksdensity} function in \texttt{\texttt{MATLAB}} produce practically identical results. We use six basis elements for the tangent space representation throughout.

For comparison, we used $100$ samples each of size $n=100$ and $n=1000$ to obtain a mean integrated squared-error loss function estimate, a mean absolute error estimate and a mean $\linf$ loss function estimate from the densities evaluated over a grid of $100$ points at $10$ equidistant locations over the support of each of the predictors. As a first example, we consider a situation where the true conditional density is 
a Laplace distribution, i.e.
$f(y_i|X=x_i)=\text{DExp}(y_i; \text{mean=}{(2x_i -1)}, \text{var=}1)$ and $X_i \sim \mathcal{N}(0,1)$. As the second example we take a bivariate predictor scenario where $f(y_i|X=(x_{1i},x_{2i}))= (1-e^{-x_{2i}}){\cal N}(y_i; (x_{1i} +2), {(0.5)}^2)  + (e^{-x_{2i}}) \text{DExp}(y_i; {(x_{1i} -1)}, 1)$ and the predictors $X_1 \sim 0.95\mathcal{N}(0,{(0.4)}^2) + 0.05\mathcal{N}(0,{(1.4)}^2)$ and $X_2 \sim \mathbb{U}(0,1)$.

The results are summarized in Table \ref{tablecde1}. From the results it is clear that when the sample size is low the performance of the warped estimate is better and more stable. When the sample size is high the performance of the two methods are more comparable though the warped estimation method still provides more stable performances. However, the computation cost of the \texttt{NP} package is very high even with the kd-tree implementation, whereas the warped estimation is computationally very efficient.  \\

\begin{table}[t!] 
\caption{\it A comparison of the performances \texttt{NP} package and Warped estimate for simulated examples.}
\label{tablecde1}\par
\vskip .2cm
\centerline{\tabcolsep=3truept\begin{tabular}{|r|r|c|rrr|rrr|} \hline 
\multicolumn{3}{|c}{Method:}
             & \multicolumn{3}{c}{NP package}
         & \multicolumn{3}{c|}{Warped Estimate} \\ \hline
\multicolumn{1}{|c}{Example} & \multicolumn{1}{c}{$n$} &  \multicolumn{1}{c}{Norm}
           & \multicolumn{1}{c}{Mean} & \multicolumn{1}{c}{std.dev} & \multicolumn{1}{c}{Time}
          & \multicolumn{1}{c}{Mean} & \multicolumn{1}{c}{std.dev} & \multicolumn{1}{c|}{Time} \\[1pt] \hline
\multirow{6}{*}{Example 1} & \multirow{ 3}{*}{100} & $\lone$ &  4.11 & 0.51 &  & 3.28 & 0.44 &    \\
& &  $ISE$ & 0.59 & 0.12 & $1$ sec & 0.41 & 0.11 & $1$ sec  \\
& &  $\linf$ &  0.40 & 0.07 &  & 0.88 & 0.34 &
   \\ \cline{2-9}
 & \multirow{ 3}{*}{1000} & $\lone$ &  2.50 & 0.24 &  & 2.46 & 0.11 &    \\
& &  $ISE$ &  0.26 & 0.04 & $51$ sec & 0.25 & 0.03 & 3 sec  \\
& &  $\linf$ & 0.39 & 0.06 &  & 0.36 & 0.04 &
\\ \hline
\multirow{6}{*}{Example 2} & \multirow{ 3}{*}{100} & $\lone$ &  60.49 & 6.67 &  & 58.55 & 5.28 &    \\
& &  $ISE$ &  11.43 & 4.01 & $2$ sec & 10.38 & 1.82 & $2$ sec  \\
& &  $\linf$ &  2.47 & 0.43 &  & 2.41 & 0.35 &
   \\ \cline{2-9}
 & \multirow{ 3}{*}{1000} & $\lone$ &  42.10 & 4.32 &  & 53.53 & 1.86 &    \\
& &  $ISE$ & 5.88 & 1.41 & $198$ sec & 8.96 & 0.57 & $7$ sec  \\
& &  $\linf$ &  2.38 & 0.29 &  & 2.24 & 0.25 &
\\ \hline
\end{tabular}}
\end{table}

\subsection{Application to Epidemiology}

\citet{longnecker2001association} studied the association of DDT metabolite DDE exposure and preterm birth in a study based on the US Collaborative Perinatal Project (CPP). DDT is very effective against malaria inflicting mosquitoes and hence is frequently used in malaria-endemic areas in spite of evidence that suggests associated health risks. Both \citet{longnecker2001association} and \citet{dunson2008kernel} concluded that higher levels of DDE exposure is associated with higher risks of preterm birth. The response variable in question is the gestational age at delivery (GAD), and deliveries occurring prior to 37 weeks of gestation is considered as preterm. \citet{longnecker2001association} also recorded the serum triglycerine level, among several other factors, and included it in their model since serum DDE level can be affected by concentration of serum lipids.

We study the Longnecker data to investigate the effect of varying levels of DDE on the distribution of GAD, focusing on the left tail of distribution to assess the effect on preterm births. In our study, following \citet{dunson2008kernel}, we include only the 2313 subjects for whom the gestation age at delivery is less than 45 weeks, attributing higher values to measurement errors. We study the conditional density of GAD given different doses of DDE in the serum. We also study the effect of different levels of triglyceride on GAD. However, since DDE is a possible confounding factor, we conduct a bivariate analysis, including both DDE dose and triglyceride level as the covariates and study the effect on GAD at varying levels of one covariate, keeping the other fixed. We also investigate whether different levels of one covariate affect the distribution of the other. 

Based on our findings, the very erratic behavior at locations where the DDE dose or triglyceride levels are 99th percentile is seen with some skepticism because of the sparsity of the data in that region. We notice an increasingly prominent peak near the left tail of GAD distribution with increasing dose of DDE, which agrees with the results of \citet{longnecker2001association} and \citet{dunson2008kernel}, shown in the left panel of Figure \ref{fig:gestvsddeandgly}. The right panel of Figure \ref{fig:gestvsddeandgly} suggests a tendency of higher risks of preterm birth at higher doses of triglycerides as well, though the difference was less pronounced.

To investigate whether the results corresponding to triglycerides were confounded by the DDE doses, we first study the effect of triglyceride levels on DDE distribution and vice versa.
Figure \ref{fig:covariate} shows that the distributions of the covariates are completely identical for varying levels of the other. The only exception is at 99th percentile of triglyceride for which the distribution of DDE doses seem to be shifted to the right. 
For fixed levels of triglyceride, increasing DDE doses shows an increasing left peak except where both DDE and triglyceride levels are very high, shown in Figure \ref{fig:glyfixed}. For fixed doses of DDE the distribution of GAD at different levels of triglyceride do not follow any increasing trend and are almost indistinguishable from each other for all the different doses of DDE, as seen in Figure \ref{fig:ddefixed}. This suggests that the increased risk of preterm birth can be attributed primarily to DDE doses, and there is no significant effect of different triglyceride levels on the gestation age.  The apparent increasing risk of preterm birth for increasing level of triglycerides seen in the right panel of Figure \ref{fig:gestvsddeandgly} is mainly caused by DDE doses acting as a confounding factor.

\begin{figure}[t!]
\begin{center}
\begin{tabular}{cc}
\includegraphics[width=3in,keepaspectratio]{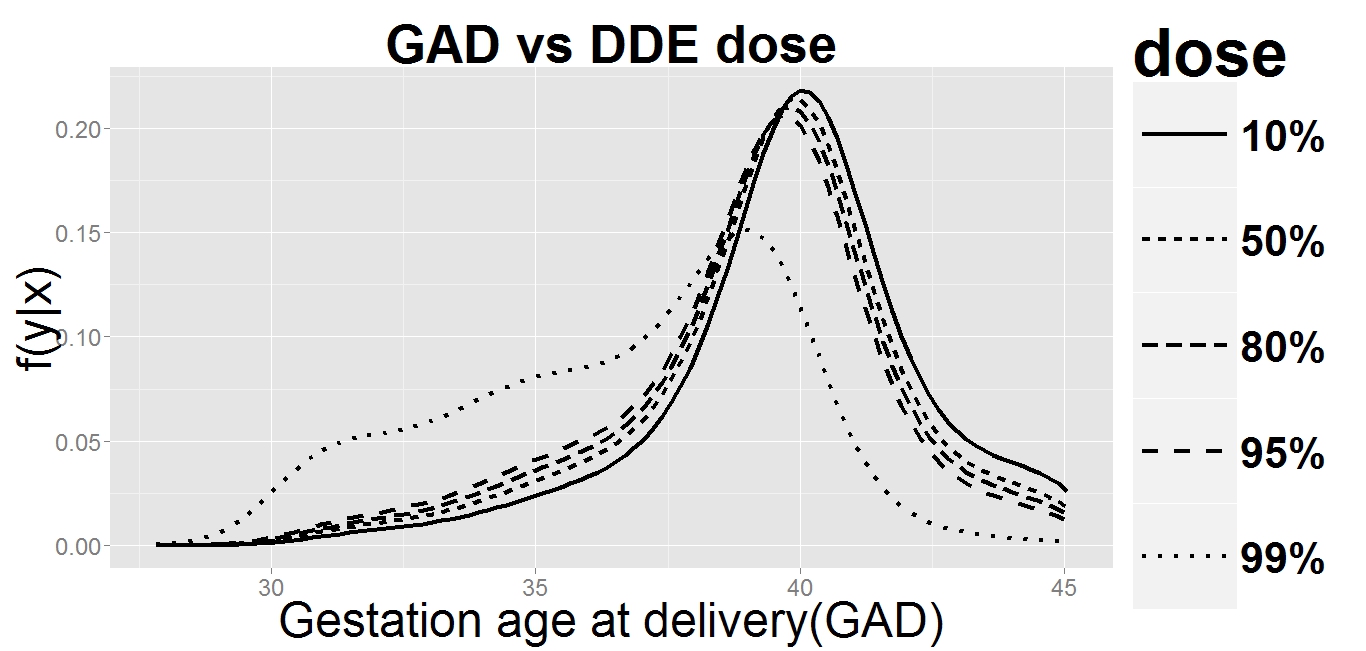} &
\includegraphics[width=3in,keepaspectratio]{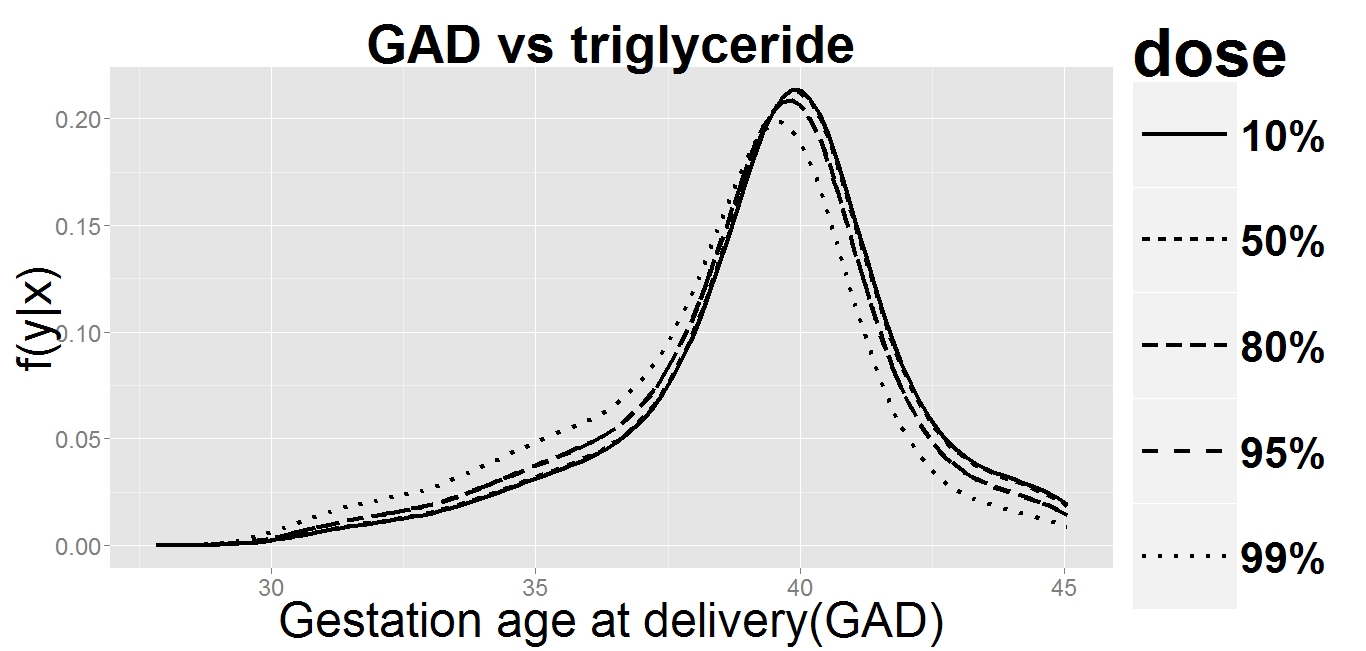} 
\end{tabular}
\caption{\it Distribution of gestation age at delivery  for varying levels of DDE and triglyceride}
\label{fig:gestvsddeandgly}
\end{center}
\end{figure}

\begin{figure}[htbp]
\begin{center}
\begin{tabular}{cc}
\includegraphics[width=3in,keepaspectratio]{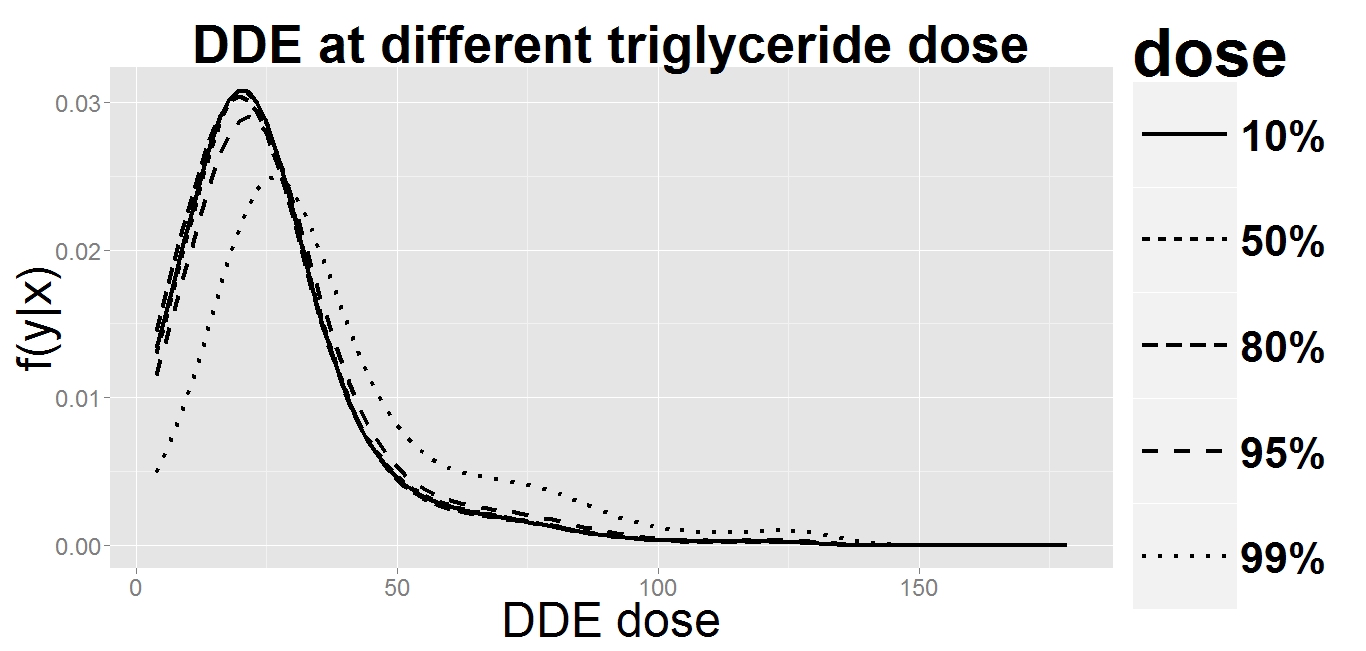} &
\includegraphics[width=3in,keepaspectratio]{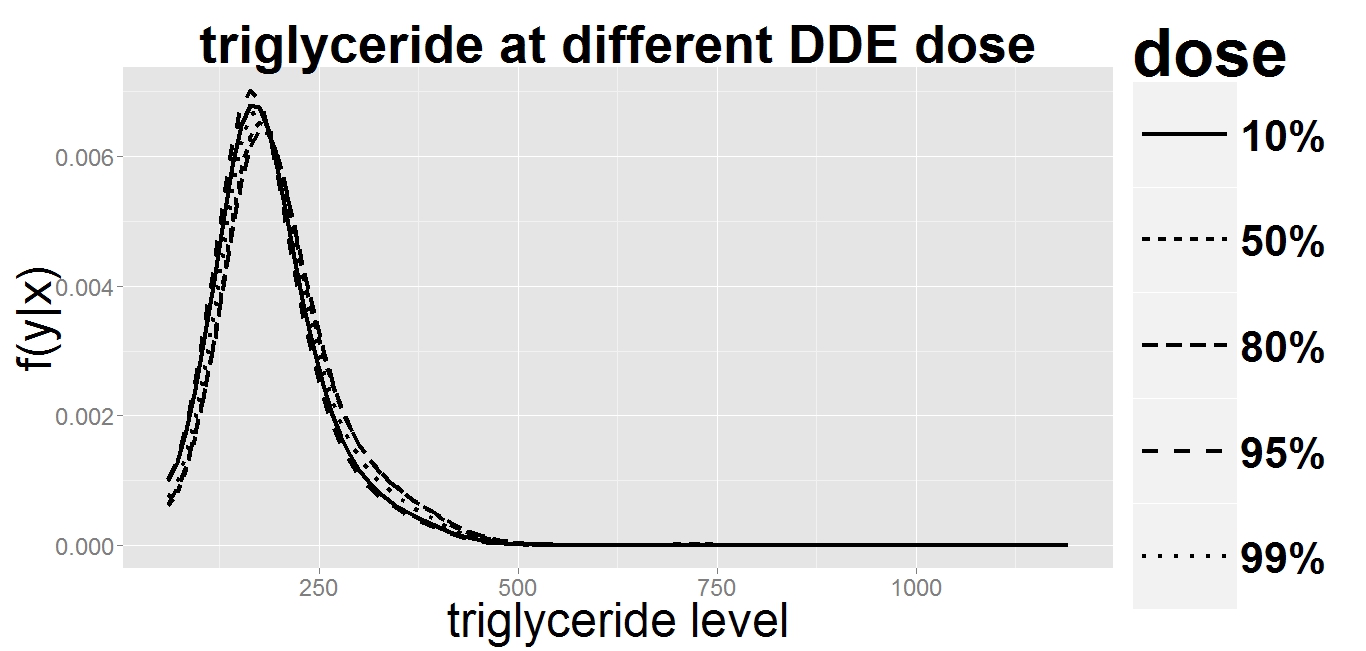} 
\end{tabular}
\caption{\it Distribution of DDE and triglyceride at different levels of the other}
\label{fig:covariate}
\end{center}
\end{figure}

\begin{figure}[htbp]
\begin{center}
\begin{tabular}{ccc}
\includegraphics[width=2in,keepaspectratio]{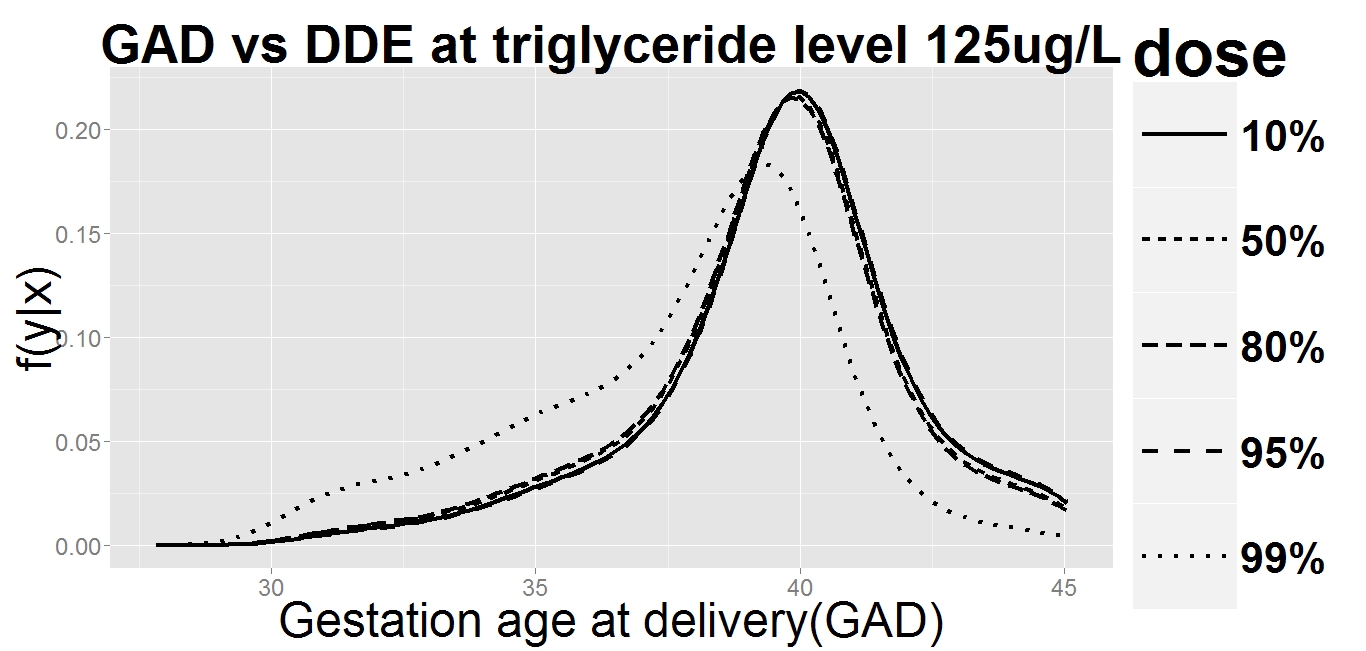} &
\includegraphics[width=2in,keepaspectratio]{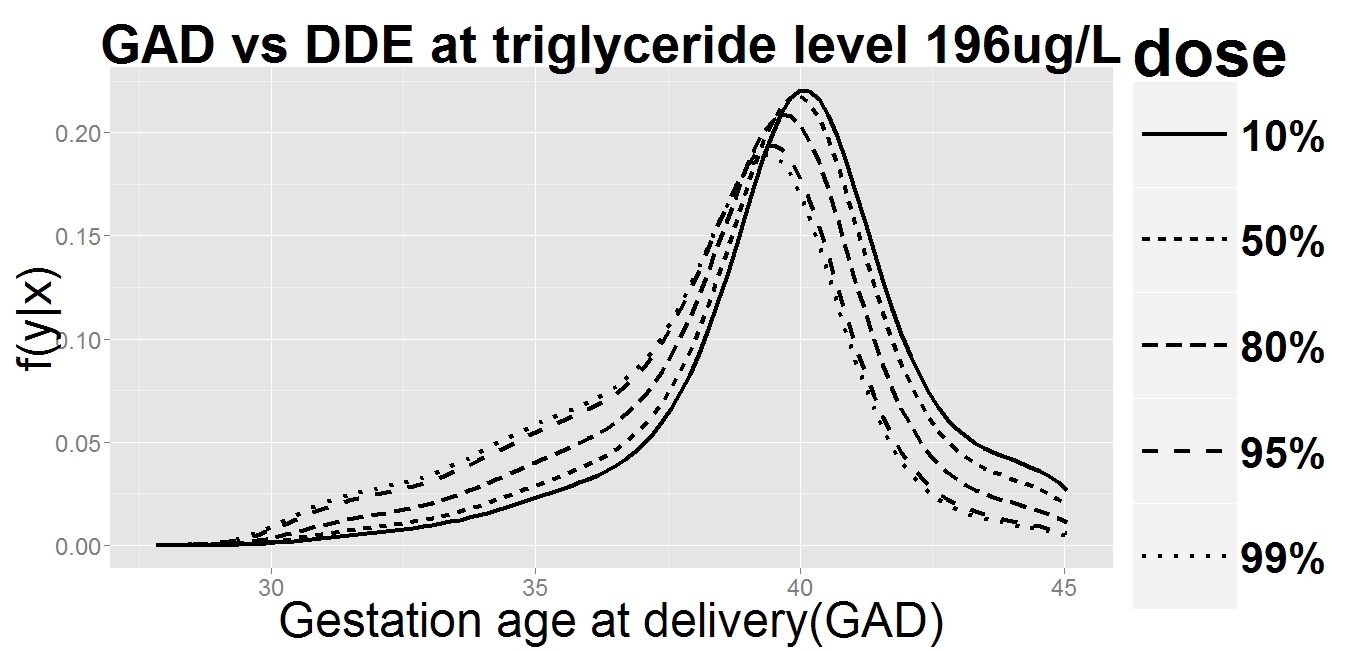} &
\includegraphics[width=2in,keepaspectratio]{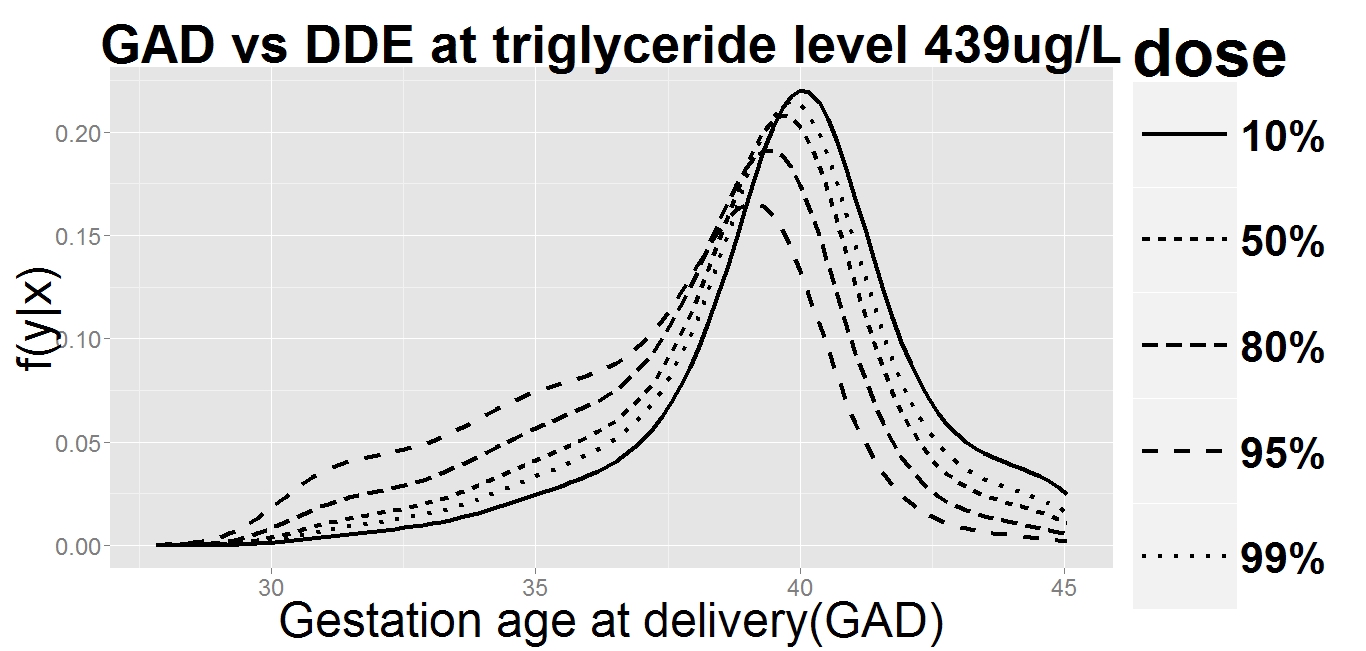} 
\end{tabular}
\caption{\it Distribution of gestation at varying levels of DDE for fixed values of triglyceride}
\label{fig:glyfixed}
\end{center}
\end{figure}

\begin{figure}[htbp]
\begin{center}
\begin{tabular}{ccc}
\includegraphics[width=2in,keepaspectratio]{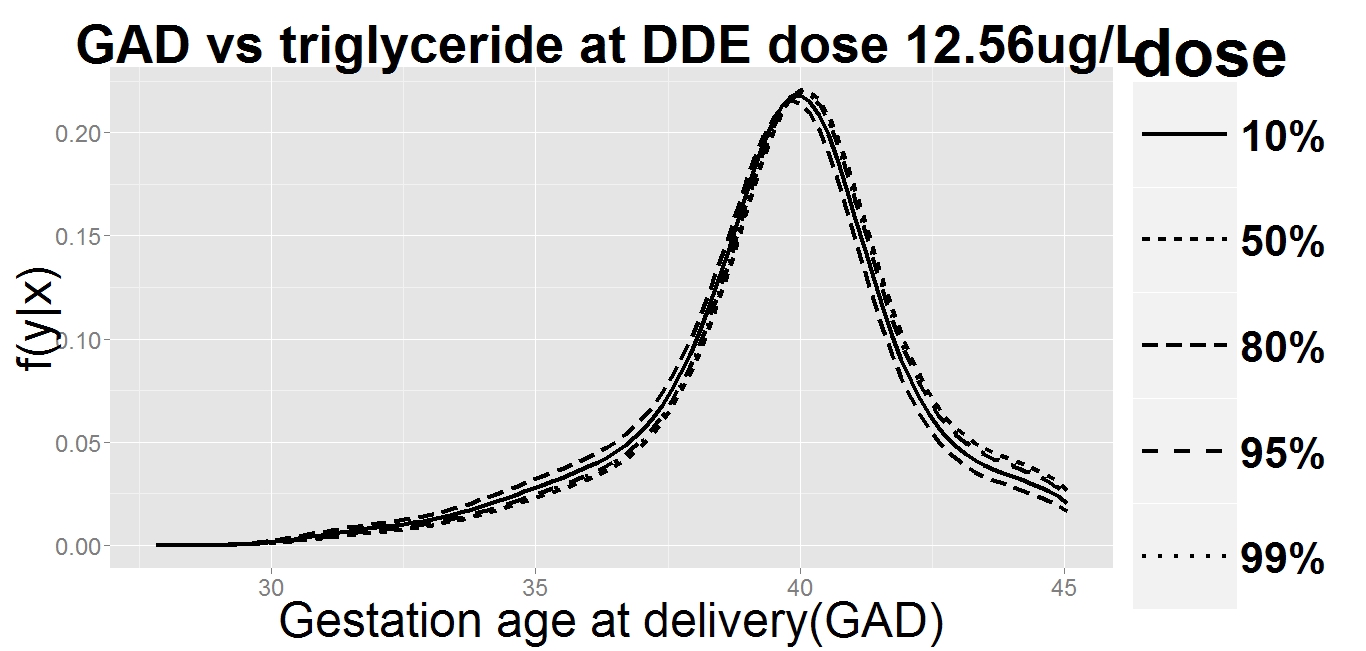} &
\includegraphics[width=2in,keepaspectratio]{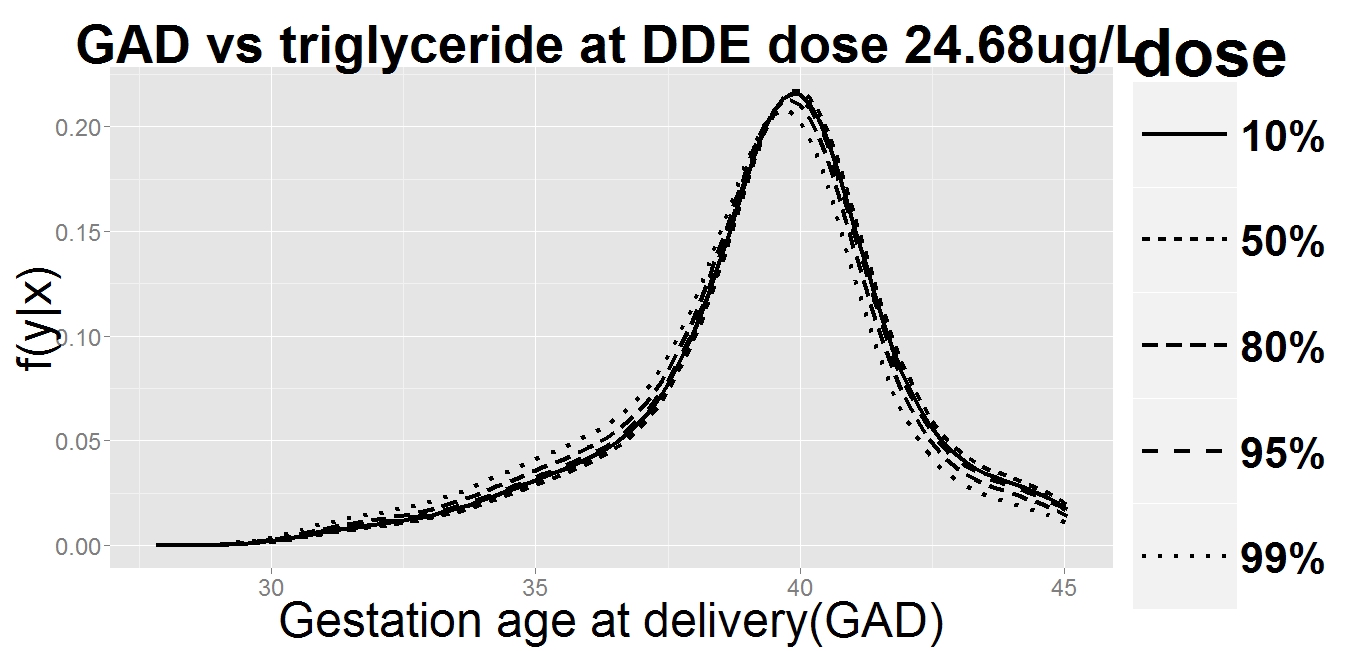} &
\includegraphics[width=2in,keepaspectratio]{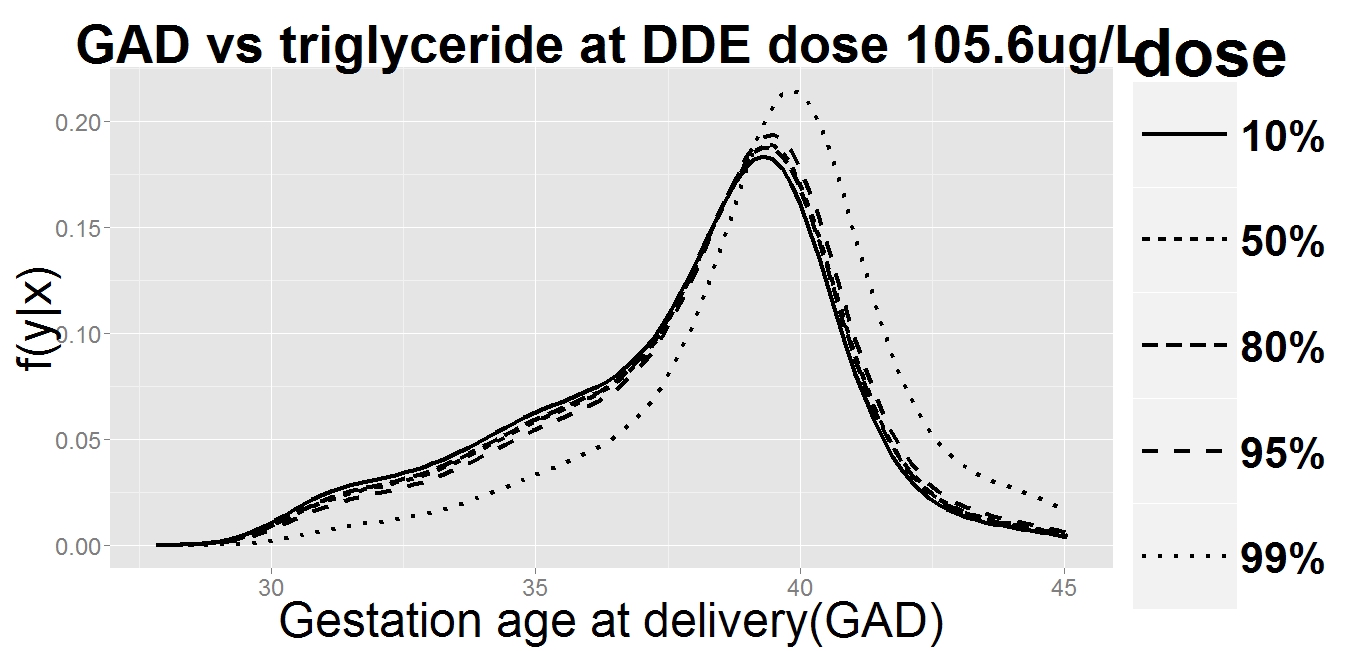} 
\end{tabular}
\caption{\it Distribution of gestation at varying levels of triglyceride for fixed values of DDE}
\label{fig:ddefixed}
\end{center}
\end{figure}

\setcounter{equation}{0} 

\newpage
\begin{centering}
{\large\bf SUPPLEMENTARY MATERIALS}
\end{centering}

\section{Theoretical Results}\label{familydisc}
Let $\mathscr{F}$ and $\mathscr{F}_n$ be as defined in Section $2$ of the manuscript.
 To control the approximation error of $\mathscr{F}_n$, \citet{wong1995probability} introduces a family of discrepancies. They define $\delta_n (p_0,\mathscr{F}_n)=\text{inf}_{f \in \mathscr{F}_n} \rho (p_0,f)$, called the $\rho$-approximation error at $p_0$. Here $p_0$ is the true density which is assumed to belong to H\"{o}lder space of order $\beta>0$ so that if $k_n=l_1n^{1/(2\beta +1)}$ then $\inf_{f \in \mathcal{F}_n} \| p_0 - f\|_\infty \leq l_2 n^{-\beta/(2\beta+1)}$ for some arbitrary constants $l_1$ and $l_2$ . The control of the approximation error of $\mathscr{F}_n$ at $p_0$ is necessary for obtaining results on the convergence rate for sieve MLEs.
We follow \citet{wong1995probability} to introduce a family of indexes of discrepency in order to formulate the condition on the approximation error of $\mathscr{F}_n$.
Let
\[
g_{\alpha} (x) = \left\{\begin{array}{lr}(1/\alpha)[x^{\alpha} -1], -1<\alpha<0 \text{ or }0<\alpha \leq 1\\
\log{x}, \text{  if }\alpha=0+
\end{array}
\right.
\]
Set $x=p/f$ and define $\rho_{\alpha} (p,f)=E_pg_{\alpha}(X)=\int pg_{\alpha}(p/f).$
We define $\delta_n (\alpha) =\inf _{f \in \mathscr{F}_n} \rho_{\alpha} (p_0,f)$. 
We call a finite set $\{(f_{j}^{L},f_{j}^{U}),j=1,\dots,N\}$ a Hellinger $u$-bracketing of $\mathscr{F}_n$ if 
${\|{f_{j}^{L}}^{1/2}-{f_{j}^{U}}^{1/2}\|}_{2} \leq u$ for $j=1,\dots,N$, and for any $p \in \mathscr{F}_n$, there is a $j$ such that $f_{j}^{L} \leq p \leq f_{j}^{U}$. 
Let $H(u,\mathscr{F}_n)$ be the Hellinger metric entropy of $\mathscr{F}_n$, defined as the cardinality of the $u$-bracketing of $\mathscr{F}_n$ of the smallest size. 
Let $f_p$ be the initial estimate on which we use the group action of the space of diffeomorphisms to arrive at the final estimate. 
Throughout, $c_1$ and $c_2 $ have been used to represent coefficient vectors  in the tangent space of the Hilbert sphere for some fixed basis set corresponding to warping function that acts on $f_p$. When $c_1$ denotes the coefficient vector corresponding to the true density denoted by $p_0 \in \mathscr{F}$ and $c_2$ corresponds to the estimate $f \in \mathscr{F}_n$, ${c_1}^{>k_n}$ represents the $(k_n +1)$th onwards coordinates of $c_1$.  $l_1, l_2, l_3$ and $l_4$ are used to indicate specific constants. Also, $M_1, M_2, M_3,\dots,$ have been used to represent generic constants whose value can change from step to step but is independent of other terms in the expressions.

\subsection{{\it pdf} space versus the coefficient space}  \label{equivproof}
Let $f_1$ and $f_2$ be two {\it pdfs} on $\mathscr{F}_n$ with corresponding cumulative distribution functions $F_1$ and $F_2$. Let $f_p$ be the initial density estimate on $\mathscr{F}_p$ such that $f_p$ is strictly positive and Lipschitz continuous with cumulative distribution function $F_p$. Let $\gamma_1={F_p}^{-1} \circ F_1$ and $\gamma_2={F_p}^{-1} \circ F_2$.
Let $c_1= (c_{11}, \ldots, c_{1k_n})^{\T}$ and $c_2  = (c_{21}, \ldots, c_{2k_n})^{\T}$  be the coefficients associated with the two elements of $T_{{\bf 1}}(\s_{\infty})$ corresponding to the tangent space representation of $\gamma_1$ and $\gamma_2$. Here  $\mathscr{F}_n$ and $k_n$ are as introduced in section $2$ of the manuscript. Then  the following Lemma bounds the norm difference of $f_1$ and 
$f_2$  with the norm difference in the coefficients.
\begin{prop}
$|f_1 -f_2| \leq M_0 \norm{c_1 - c_2}_{1}$
where $M_0>0$ is a constant.
\end{prop}

\begin{proof}
Let $c_1$ and  $c_2$ be the coefficients associated with two elements $v_1$ and $v_2$ of $T_{{\bf 1}}(\s_{\infty})$, defined in Section $2$ of the manuscriptand let $q_1$ and $q_2$ represent the corresponding elements on the Hilbert sphere.  Then there exists $  M_1 \in \real$ such that $|B_{i}| < M_1 $ , where $B_{i}$ is the $i$th basis function, $i=1,2,\cdots,k_n$. Let 
$
v_1 =\sum_{i=1}^{k_n} c_{1i}B_{i},$ $ v_2 =\sum_{i=1}^{k_n} c_{2i}B_{i}.  
$
Then $v_1,v_2 \in T_{{\bf 1}}(\s_{\infty})$  with $\|v_1\|<\pi/4$ and $\|v_2\|<\pi/4$.  Hence we have 
\begin{eqnarray*}
(v_1 - v_2 ) (t) =\sum_{i=1}^{k_n}(c_{1i} -c_{2i} ) B_{i}(t) <M_1 \sum_{i=1}^{k_n}| c_{1i} - c_{2i} |=M_1{\|c_1 -c_2\|}_1
\end{eqnarray*}
\begin{eqnarray*}
 \|v_1 - v_2\|=\sqrt{\int_0^1 {(v_1 -v_2)}^{T} (v_1 - v_2) dt} <M_3\sqrt{ \sum_{i=1}^{k_n} {(c_{1i} -c_{2i} )}^2}<M_1  {\| c_1 - c_2 \|}_{1}
\end{eqnarray*}

Next since $x \mapsto \|x\|=\sqrt{\int_0^1 x^2(t)dt}$ and $x \mapsto \cos(x)$ are Lipschitz continuous, we have
\begin{equation}
| \cos \|v_1\| - \cos \|v_2\| | < M_{2}| \|v_1\| - \|v_2\|| < M_1  {\| c_1 - c_2 \|}_{1}
\label{eq:cosabs}
\end{equation}
Next note that $x \mapsto \sin(x) /x $ is Lipschitz continuous. Hence we have
\begin{equation}
\norm{\frac{\sin \|v_1\|}{\|v_1\|} - \frac{\sin \|v_2\|}{\|v_2\|}} < M_{2}|\|v_1\| - \|v_2\| |< M_1 {\| c_1 - c_2 \|}_{1}
 \label{eq:sinxbyx}
\end{equation}

Noting that 
\begin{eqnarray*}
|q_1 (t)-q_2(t)| <\abs{\cos\|v_1\| - \cos \|v_2\|} +\abs{\frac{\sin \|v_1\|}{\|v_1\|} v_1 (t)- \frac{\sin\|v_2\|}{\|v_2\|}v_2(t)}
\end{eqnarray*} we have, combining equations \ref{eq:cosabs} and \ref{eq:sinxbyx},
\begin{equation}
\norm{q_1 -q_2}_{1} < M_1 \norm{c_1 - c_2}_{1}
\label{gamdot}
\end{equation}
Now consider $Q=q^2$. Observe that 
\begin{eqnarray*}
(Q_1 -Q_2)(t) &=& {q_1}^{2} (t)- {q_2}^{2}(t)=(q_1(t) - q_2(t))(q_1(t) + q_2(t))\\
&=& (\cos |v_1| + \cos |v_2| + \frac{\sin |v_1|}{|v_1|} v_1 (t) + \frac{\sin|v_2|}{|v_2|}v_2(t))(q_1(t) - q_2(t)).
\end{eqnarray*}
Now   $(\cos \|v_1\| + \cos \|v_2\| + \frac{\sin \|v_1\|}{\|v_1\|} v_1 (t) + \frac{\sin \|v_2\|}{\|v_2\|}v_2(t))$ is a bounded function. Hence ${\|Q_1 -Q_2\|}_{1} < M_1\norm{c_1 - c_2}_{1}$ using equation \ref{gamdot}.  Now we have $\gamma_i(t)= \int_{0}^{t} Q_i(u) du$, $t \in [0,1],i=1,2$. Then 
\begin{eqnarray*}
|\gamma_1 (t) -\gamma_2 (t)| = \abs{\int_{0}^{t}\bigg(Q_1(u) - Q_2 (u)\bigg)du} < \int_{0}^{t} |Q_1(u) - Q_2 (u) | du  \leq \norm{Q_1 - Q_2}_{1}
\end{eqnarray*}
 
Since $f_p$ is Lipschitz continuous and strictly positive density on $[0,1]$, we have
\[
{\|f_p (\gamma_1) -f_p (\gamma_2) \|}_{1} <M_4 {\|\gamma_1 -\gamma_2\|}_{1}
\]
Consider $|f_1 - f_2| =|f_p (\gamma_1). \dot{\gamma}_1- f_p (\gamma_2). \dot{\gamma}_2|$. 
Keeping in mind that $Q=\dot{\gamma}$, we have
\begin{eqnarray*}
|f_1(t)-f_2(t)| &=& |f_p(\gamma_1 (t)).Q_1(t)-f_p(\gamma_2 (t)).Q_2(t)| \\
&=& |f_p(\gamma_1 (t)).Q_1(t)-f_p(\gamma_2 (t)).Q_1(t)+f_p(\gamma_2 (t)).Q_1(t)-f_p(\gamma_2 (t)).Q_2(t)| \\
&\leq& |Q_1(t)|M_1{\|\gamma_1 - \gamma_2\|}_{1} + |f_p(\gamma_2(t))|{\|Q_1 - Q_2\|}_{1}\\
&\leq & M_{2}{\|\gamma_1 - \gamma_2\|}_{1} + M_{3}{\|\gamma_1 - \gamma_2\|}_{1} <M_0 \norm{c_1 - c_2}_{1}.
\end{eqnarray*}

Therefore we have $|f_1 - f_2|< M_0 \norm{c_1 - c_2}_{1}$ for some fixed $M_0>0$.
\end{proof}

{\bf Remark 1}:$H(f_1,f_2) <M_1\sqrt{{\|f_1-f_2\|}_1} <M_1\sqrt{ \norm{c_1 - c_2}_{1}}<l_1\sqrt{ \norm{c_1 - c_2}_{\infty}}$ for some fixed $l_1>0$ where $H(f_1,f_2)$ is the Hellinger metric between two densities $f_1$ and $f_2$.

\subsection{Proof of Lemma 1 and Corollary 1}
Let us consider a fixed $f_0 =f_p (\Gamma(c_0)). \dot{\Gamma}(c_0)$.
We note that $H(f_1,f_2)\leq l_1 \sqrt{{\|c_1 -c_2\|}_{\infty}}$ for some $l_1>0$ following the steps in section \ref{equivproof}.
So finding a $\delta$ covering for $\mathscr{F}_n$ is equivalent to finding an $l_1\sqrt{\delta}$ covering for the space of coefficients in the tangent space using $L_\infty$ norm.  Let us have a closer look at the space of coefficients. We have $\|v\| < \pi/4$ for tangent space representation of $\Gamma$, which is equivalent to ${\|c\|}_2 \leq l_3$,say.
Therefore $\mathscr{F}_n \equiv \{c \in {\mathbb{R}}^{k_n}:{\|c\|}_2 \leq l_3\}=\mathscr{C}$,say. 
Then $\mathscr{C} \subset \{c \in {\mathbb{R}}^{k_n}:{\|c\|}_\infty \leq l_4\} \equiv \{c \in {\mathbb{R}}^{k_n}:{|c_i|} \leq l_4 \forall i=1,\dots, k_n\}=\mathscr{C}_0$,say.
Now $\mathscr{C}_0$ is a compact set with $\mathscr{C}$ as a compact subset. Therefore the covering number N for $\mathscr{C}$ would be less than the covering number for $\mathscr{C}_0$.
Since $\mathscr{C}_0 \equiv \{{[-l_4,l_4]}^{k_n}\}$, we have
 the covering number for $\mathscr{C}_0$ as ${(\frac{2l_4}{l_1\sqrt{\delta}})}^{k_n}$. We obtain this by partitioning the interval$[-l_4,l_4]$ into pieces of length $l_1\sqrt{\delta}$ for each coordinate so that the partition of $\mathscr{C}_0$ is reached through cross product. Then in each equivalent class of the partition of  $\mathscr{C}_0$ we will have ${\|c_1 -c_2\|}_{\infty}\leq l_1\sqrt{\delta}$ which is equivalent to $H(f_1 ,f_2)\leq \delta$.
So we have the metric entropy for $\mathscr{F}_n=H(.,\mathscr{F}_n)=H(u,\mathscr{F}_n)<k_n \log{l/u}$, where $l=2l_4$ and $u=l_1\sqrt{\delta}$.
Now,
\[
\int_{{\epsilon}^2/2^8}^{\sqrt{2}\epsilon} {H}^{1/2} (\frac{u}{l_3},\mathscr{F}_n) du\leq \sqrt{k_n}\int \sqrt{\log(l_0/u)} du \leq \sqrt{k_n\log(M/{\epsilon}^2)} (\sqrt{2}\epsilon -{\epsilon}^2/256)
\]
where $l_0=l_3l$ and $M=2^8l_0$.
For the existence of an $\epsilon_n$ that satisfies Lemma $1$ we need an $\epsilon_n$ less than $1$ that satisfies
\begin{equation}
\sqrt{k_n\log(M/{\epsilon}^2)} (\sqrt{2}\epsilon -{\epsilon}^2/256) \leq C_4 n^{1/2} {\epsilon}^2
\label{satisfiesprop}
\end{equation}

But this inequality holds at $1-$ and hence there exists a smallest $\epsilon_n <1$ that satisfies \ref{satisfiesprop}. The corollary follows directly from Theorem $1$ in  \citet{wong1995probability}

\subsection{Proof of Lemma 2}
%

Consider $\alpha=1$ in \eqref{eq:main}. $\delta_n (1) =\inf _{f \in \mathscr{F}_n} \rho_{1} (p_0,f)=\inf _{f \in \mathscr{F}_n} \int p_0g_{1}(p_0/f).=\inf _{f \in \mathscr{F}_n} \int \frac{{(p_0 -f)}^2}{f}$.  Let $P_0$ and $F_p$ be the cdfs corresponding to the true density and the initial parametric estimate respectively. Then we have $\gamma_0={F_p}^{-1}\circ P_0 $ has the tangent space representation $v_0$ obtained via exponential map of $\sqrt{\dot{\gamma_0}}$ satisfying $\|v_0\| <\pi/4$. This forces $ \cos (\|v_0\|) + \frac{\sin(\|v_0\|)}{\|v_0\|}v_0$ to be always positive. Let $f$ be the final density estimate and $c_2$ be the corresponding coefficient vector in the tangent space representation and $v_1$ be the corresponding element in the tangent space.
Now  we have $\|v_1\|<\pi/4$ corresponding to $f$ because $c_2 \in V_{\pi}^{k_n}$ following the notation $V_{\pi}^{J}$ introduced in Section $2$ of the manuscript.
That implies $\cos (\|v_1\|) + \frac{\sin(\|v_1\|)}{\|v_1\|}v_1>0$, i.e. 
  \[
\dot{\gamma}(t)={\big( \cos (\|v_1\|) + \frac{\sin(\|v_1\|)}{\|v_1\|}v_1}\big)^2(t) >0\text{  } \forall t\in [0,1]
 \]
 Also $\dot{\gamma}(t)$ is continuous in $t$ on a closed and bounded interval. So it attains its minima at some point $t_0$ such that $\dot{\gamma}(t) \geq \dot{\gamma}(t_0) > 0$  for all  $t\in [0,1]$.
 Thus it follows that  $f(t) > M_1 \dot{\gamma}(t_0) = d,\text{say}$.    Then we have, $\delta_n (1) = \inf _{f \in \mathscr{F}_n} \int \frac{{(p_0 -f)}^2}{f} dt < \inf _{f \in \mathscr{F}_n} {\|p_0 - f\|}_{\infty}^2/d = C_5 n^{\frac{-2\beta}{2\beta +1}}$ for some $C_5>0$.  

\subsection{Proof of Theorem 1}
 
We have from equation \ref{satisfiesprop} $\sqrt{k_n\log(M/{\epsilon}^2)} (\sqrt{2}\epsilon -{\epsilon}^2/256) < \sqrt{k_n\log(M/{\epsilon}^2)}\sqrt2 \epsilon$.
So for an upper bound of the smallest root we can solve the equation  $\sqrt{k_n\log(M/{\epsilon}^2)}\sqrt2 \epsilon=C_4 n^{1/2} {\epsilon}^2$.   
Let $\epsilon_n$ be of the form $\sqrt{M}n^{-\gamma}{(\log n)}^t,\gamma >0$, and, let $k_n=n^{\Delta}$, $\Delta <1$  
Then 
$\log{(M/{{\epsilon_n}^2})}=2\gamma \log{n} -2t \log{\log{n}} \leq 2\gamma\log{n}$.

So for an upper bound of the smallest root we can solve the equation 

  $\sqrt{k_n 2\gamma\log{n}} \sqrt{2} \epsilon=C_4 n^{1/2} {\epsilon}^2$.   
  
Therefore equating, $n^{\Delta /2}\sqrt{2\gamma}\sqrt{\log{n}}\sqrt{M}n^{-\gamma}{(\log{n})}^t$ with $C_4Mn^{1/2}n^{-2\gamma}{(\log{n})}^{2t}$,
 we get $\gamma=\frac{1}{2} (1-\Delta)$, and $t=1/2$.
 Thus we have $\epsilon_n = \sqrt{M}n^{\frac{-(1-\Delta)}{2}}\sqrt{\log n}$.
 We take $\Delta$ to be $\frac{1}{2\beta + 1}$ to use the theoretical properties of H\"{o}lder space of order $\beta>0$.  
 Therefore $\epsilon_n = \sqrt{M}n^{\frac{\beta}{2\beta+1}}\sqrt{\log n}$ is an upper bound for the smallest value that satisfies the condition for Lemma $1$.   Therefore, using the definition given in Theorem $4$ in \citet{wong1995probability}, and using $\alpha=1$, we get
\[
\epsilon_n^*  =\left\{\begin{array}{lr}Mn^{-\beta/(2\beta +1)}\sqrt{\log n}, \text{   if  } \delta_n (1) < \frac{1}{4}C_1M^2n^{-2\beta/(2\beta +1)}\log n,\\
{(4\delta_n(1)/C_1)}^{1/2},\text{ otherwise.}
\end{array}
\right.
\]
But $\delta_n (1)=C_5n^{-2\beta/(2\beta +1)} < \frac{1}{4}C_1M^2n^{-2\beta/(2\beta +1)}\log n$ for $n > \exp(4C_5/M^2C_1)$. Thus for large enough $n$, $\epsilon_n^*=Mn^{-\beta/(2\beta +1)}\sqrt{\log n}$ and following Theorem $4$ of \citet{wong1995probability} we get 
\begin{eqnarray}\label{eq:main}
 P({\|q^{1/2}-p_{0}^{1/2}\|}_2 \geq \epsilon_n^* ) \leq 5\text{exp}\big(-C_2n{(\epsilon_n^* )}^2\big) + \text{exp}\big(-\frac{1}{4}n\alpha C_1{(\epsilon_n^* )}^2\big).  
\end{eqnarray}

\section{Estimation Algorithm}

In this section we outline the estimation procedure and discuss some of the
implementation issues. We discretize density functions using a dense uniform partition,  $T=100$ equidistant points  over the interval $[0,1]$.
For approximating derivatives of a function, 
for example $\dot{\gamma}$ for a warping function $\gamma$, we use the first-order differences. The integrals are approximated 
using the trapezoidal method.

For optimizing log-likelihood function according to Equation 2.5 of the manuscript,  we
use the function {\it fminsearch} in \texttt{MATLAB} for our experiments. 
The {\it fminsearch} function uses a very efficient grid search technique to find the optimal values of 
coefficients $\{c_j\}$, corresponding to the chosen basis elements, to approximate the optimal warping function $\gamma$. 
However, {\it fminsearch} function can get stuck in locally-optimal solutions
in some situations. To alleviate this problem we use an iterative,  multi-resolution approach as follows. We start the optimization using a 
small number of basis elements $J$ with $c = {\bf 0}$, the point that maps to $\gamma_{id} \in \Gamma$ under $H^{-1}$. This implies a low-resolution search and 
low-dimensional search space $\real^J$. Then, at each successive iteration we increase the resolution by increasing $J$ and 
use the previous solution as the initial condition (with the additional components
set to zero) for the next stage. This slow increase in $J$, while continually  improving the optimal point $c$, performs much better
in practice than using a large value of $J$ directly in {\it fminsearch}.

Another important numerical issue is the final choice of $J$. 
For a fixed sample of size $n$, a large value of $J$ may lead to overfitting and $\hat{f}$ being a rough function. 
Also, a large value of $J$ makes it harder for the search procedure to converge to an optimal solution. 
\citet{efromovich2010orthogonal} and the references there in discusses different data-driven methods to choose the number of basis elements, by considering the number of basis elements itself as a parameter. We take a different  data-driven approach for selecting the desired number of basis elements.  
Using a predetermined maximum number of basis points, we navigate through increasing number of basis elements and at each step, we compute the value of the Akaike's Information Criterion (AIC) and choose the number of basis elements that results in the best value of the AIC, penalizing the number of basis functions used.   We summarize the full procedure in {\bf Algorithm 1}.
\begin{algorithm}
\caption{Improving solutions using {\it fminsearch} by tweaking the starting points}
i. Start with a low number of basis elements, say $J$ \\
ii. Use {\textbf 0} vector as the starting point and find the solution {\textbf d} using {\it fminsearch}.\\
iii. Increase the number of basis elements, say $J_1$ more basis elements.\\
iv. Use {\bf [0,0]} and {\bf [d,0]} as two starting points. Compare the AIC for the two cases and choose the solution with better AIC value. Call the solution {\textbf d} the optimal solution.\\
v. If the number of basis elements exceeds a predetermined large number, stop. Else go to step iii.
\end{algorithm}

Experimental results show that Bayesian Information Criterion (BIC) overpenalizes the number of basis elements used and, therefore, some 
sharper features of the true density are lost in the estimate. So 
the experiments presented in the following sections use only the AIC penalty.
\label{alg}
\section{Simulation Studies}
Next, we elaborate on the results from experiments on univariate unconditional density estimation procedure 
involving two simulated datasets, from Section 5 in the manuscript. 
The computations described here
are performed on an Intel(R) Core(TM) i7-3610QM CPU processor laptop, and the 
computational times are reported for each experiment. 
We compare the proposed solution with two standard techniques: (1) 
 kernel density estimates with bandwidth selected by unbiased cross validation method, henceforth referred to as {\it kernel(ucv)}, 
 (2) a  standard Bayesian technique using the  function {\it DPdensity} in the R package \texttt{DPPackage}. 
 The Bayesian approach naturally has a longer run-time.  For both the simulated examples, we use $2000$ MCMC runs with $500$ iterations as burn in period for the Bayesian technique.
We compare the methods both in terms of numerical performance and computational cost. 
Here we illustrate the performance of the various methods using a representative simulation. 
We highlight the performance improvement over an (misspecified) initial parametric and nonparametric density estimate brought about by warping. 
For the initial parametric estimate we have chosen a normal density truncated to $[0,1]$ with mean and standard deviation estimated from the sample. For the initial nonparametric estimate, we used inbuilt \texttt{MATLAB} function {\it ksdensity}.

\subsection{Example 1}
We borrow the first example from \citet{tokdar2007towards} and \citet{lenk1991towards},
where $p_0 \propto 0.75 \text{exp}(\text{rate}=3) + 0.25 {\cal N}(0.75,2^2)$, 
a mixture of exponential and normal density truncated to the interval $[0,1]$: 
We generate $n=100$ observations to study estimation performance. 
Here we use Meyer wavelets as the basis set for the tangent space representation of $\gamma$s. We use an ad hoc 
choice of $J=15$ basis elements to approximate the tangent space. Also, we  use an unpenalized log likelihood for optimization.

\begin{figure}[h]
\begin{center}
\begin{tabular}{cc}
\includegraphics[width=2.7in]{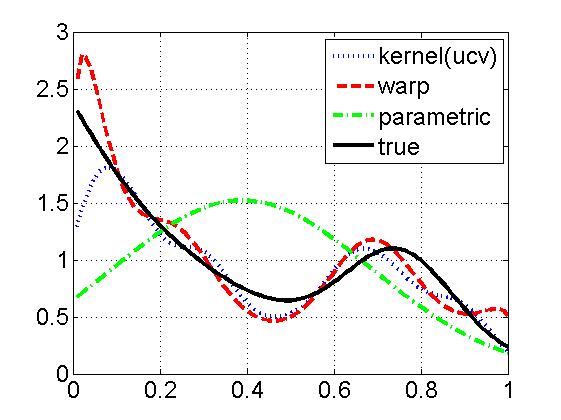} &
\includegraphics[width=2.7in]{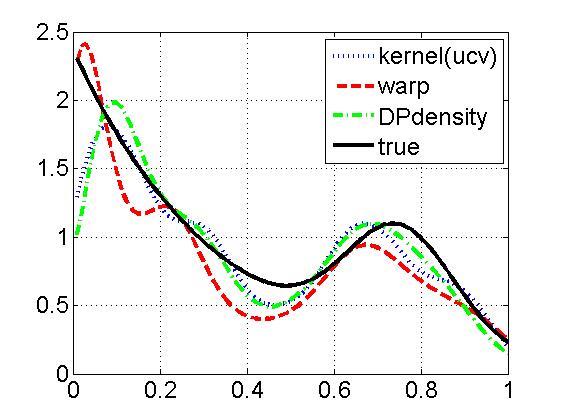} 
\end{tabular}
\caption{\it The left panel compares the warped estimate $\hat{f}$ with other estimates when $f_p$ is parametric. 
The middle panel shows the corresponding evolution of the negative of log-likelihood function during 
optimization. The right figure compares the warped estimate with others when $f_p$ is {\it ksdensity}.}
\label{fig:tokdar}
\end{center}
\end{figure}
Figure \ref{fig:tokdar} (left panel) shows a substantial improvement in the final warped estimate over the initial parametric estimate.
Incidentally, it also does a better job in capturing the left peak as compared to the {\it kernel(ucv)} method. Standard kernel methods need additional boundary correction techniques to be able to capture the density at the boundaries, as studied in \citet{karunamuni2008some} and the references therein. However the warped density seems to perform better estimation near the boundaries compared to the other techniques.
The right panel displays 
the warped result when using 
{\it ksdensity} output as the initial estimate. It also provides solutions obtained using 
{\it kernel(ucv) } and {\it DPdensity}. 
Once again,  this warped estimate provides a substantial improvement over the initial solution.

\subsection{Example 2}
For the second example we take Example 10 from \citet{marron1992exact}, which uses 
a claw density: $p_0  = \frac{1}{2}{\cal N}(0,1) + \sum_{l=0}^{4} \frac{1}{10}{\cal N}(\frac{l}{2} -1, {(0.1)}^2)$. 
\begin{figure}[h!]
\begin{center}
\begin{tabular}{cc}
\includegraphics[width=2.6in]{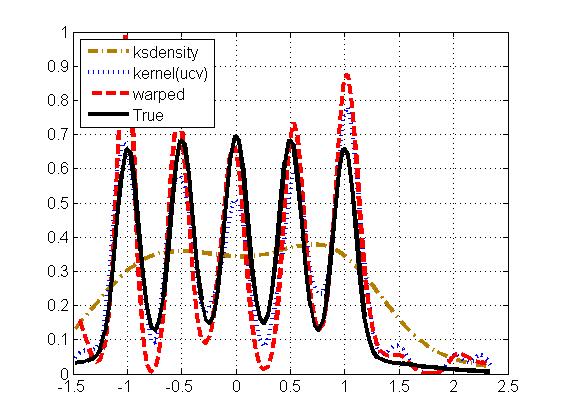} &
\includegraphics[width=2.6in]{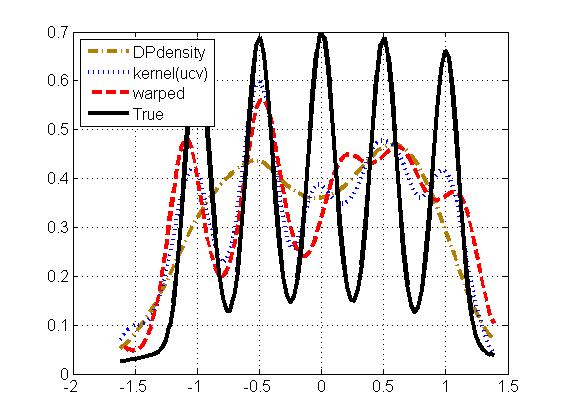}
\end{tabular}
\caption{\it The left panel shows the improvement over initial {\it ksdensity} estimate. Both {\it kernel(ucv)} and warped estimate have a good performance here. The right panel shows that all the methods fail to capture all the peaks. {\it Kernel(ucv)} performance is very similar to the warped estimate.}
\label{fig:clawpic}
\end{center}
\end{figure}
We estimate the domain boundaries and unlike the previous example, instead of fixing the number of tangent basis elements, 
we employ Algorithm 1 described in Section \ref{alg} to find the optimal number of basis elements based on the AIC, with a maximum allowed value of $40$ basis elements. 
Consequently,  the computation cost goes up.
\par
\vskip 14pt
\noindent {\large\bf Acknowledgements}

This research was supported in part by the NSF grants to AS --  NSF DMS CDS\&E 1621787 and NSF CCF 1617397.
\par


\bibhang=1.7pc
\bibsep=2pt
\fontsize{9}{14pt plus.8pt minus .6pt}\selectfont
\renewcommand\bibname{\large \bf References}
\expandafter\ifx\csname
natexlab\endcsname\relax\def\natexlab#1{#1}\fi
\expandafter\ifx\csname url\endcsname\relax
  \def\url#1{\texttt{#1}}\fi
\expandafter\ifx\csname urlprefix\endcsname\relax\def\urlprefix{URL}\fi

\bibliographystyle{plainnat}
\bibliography{densityref}

\vskip .65cm
\noindent
Florida State University
\vskip 2pt
\noindent
E-mail: (s.dasgupta@stat.fsu.edu)
\vskip 2pt

\noindent
Texas A\&M University
\vskip 2pt
\noindent
E-mail: (debdeep@stat.tamu.edu)
\vskip 2pt

\noindent
Florida State University
\vskip 2pt
\noindent
E-mail: (anuj@stat.fsu.edu)
\vskip 2pt
\end{document}